\documentclass[journal, onecolumn, 11pt]{IEEEtran}
\usepackage{amsmath}
\usepackage{amssymb}
\usepackage{amsfonts}
\usepackage{graphicx}
\usepackage{epsfig}
\usepackage{subfigure}
\usepackage{psfrag}
\usepackage{algorithmic}
\usepackage{algorithm}

\usepackage{xspace}
\usepackage{mathrsfs}
\usepackage{amsopn}

\newcommand{\Rc}{\mathcal{R}}

\newcommand{\Wc}{\mathcal{W}}

\newcommand{\st}{{\tilde{s}}}

\newcommand{\wt}{{\tilde{w}}}

\DeclareMathOperator\E{E}

\def\textiid{i.i.d.\@\xspace}
\newcommand\iid{\ifmmode\text{ i.i.d. } \else \textiid \fi}

\newtheorem{theorem}{Theorem}
\newtheorem{proposition}{Proposition}
\newtheorem{corollary}{Corollary}
\newtheorem{remark}{Remark}[section]
\newtheorem{claim}{Claim}
\title{Energy Cooperation in Cellular Networks with Renewable Powered Base Stations}

\author{Yeow-Khiang Chia\IEEEauthorrefmark{1}, Sumei Sun\IEEEauthorrefmark{2} and Rui Zhang\IEEEauthorrefmark{3}
\thanks{Paper presented in part at IEEE Wireless Communications and Networking Conference 2013, Shanghai, China} \thanks{\IEEEauthorrefmark{1} Institute for Infocomm Research, Singapore. Email: chiayk@i2r.a-star.edu.sg}  \thanks{\IEEEauthorrefmark{2}  Institute for Infocomm Research, Singapore. Email: sunsm@i2r.a-star.edu.sg}
\thanks{\IEEEauthorrefmark{3} Institute for Infocomm Research, Singapore and National University of Singapore. Email: elezhang@nus.edu.sg}%
}

\begin{document}

\maketitle \thispagestyle{empty}

\begin{abstract}
In this paper, we propose a model for energy cooperation between cellular base stations (BSs) with individual hybrid power supplies (including both the conventional grid and renewable energy sources), limited energy storages, and connected by resistive power lines for energy sharing. When the renewable energy profile and energy demand profile at all BSs are deterministic or known ahead of time, we show that the optimal energy cooperation policy for the BSs can be found by solving a linear program. We show the benefits of energy cooperation in this regime. When the renewable energy and demand profiles are stochastic and only causally known at the BSs, we propose an online energy cooperation algorithm and show the optimality properties of this algorithm under certain conditions. Furthermore, the energy-saving performances of the developed offline and online algorithms are compared by simulations, and the effect of the availability of energy state information (ESI) on the performance gains of the BSs' energy cooperation is investigated. Finally, we propose a hybrid algorithm that can incorporate offline information about the energy profiles, but operates in an online manner.
\end{abstract}

\section{Introduction}
In recent years, motivated by environmental concerns and energy cost saving considerations, telecommunication service providers have started considering the deployment of renewable energy sources, such as solar panels and wind turbines, to supplement conventional power in powering base stations (BSs). In some places where the conventional power grid is still under-developed, the deployment of renewable energy sources is more attractive due to the significantly higher costs, as compared to a developed city, in powering BSs using conventional power sources. Examples where such a scenario occurs include the deployment of BSs with renewable energy sources by Ericsson in Africa~\cite{Ericsson} and Huawei in Bangladesh~\cite{Huawei}.

Although renewable energy sources are attractive for the above reasons, they also suffer from significantly higher variability as compared to conventional energy sources. As a result, even in BSs that deploy renewable energy sources, conventional energy sources, such as diesel generators or the power grid, is still required to compensate for the variability of the renewable energy sources. One practical method of mitigating the variability of renewable energy sources is through energy storage means such as fuel cells and batteries. Energy storage, however, is very costly to deploy and therefore, the amount of storage available at BSs will usually be quite limited. A key consideration in deploying BSs with renewable energy sources is minimization of the amount of conventional energy consumed, because it is only then cost-effective to deploy renewable energy sources and storage. A survey of issues involved in deploying renewable energy sources in BSs is given in~\cite{Hasan2011}. Related work on deploying renewable energy sources in smart grids, not necessarily constrained to a communications systems setup can be found in~\cite{Kanoria2011},~\cite{Su2011} and the references therein. 

Other than work in smart grid, a closely related area of research is in the area of energy harvesting for wireless communications, where several authors have proposed the idea of energy cooperation between different nodes in a communications network; see e.g.~\cite{Gurakan--Ozel--Yang--Ulukus2012a}, and~\cite{CHuang2012}. More broadly, the area of ``green communications'' has attracted significant attention from the communications community in recent years. For an overview of the many significant research activities in this area, interested readers may refer to, e.g.,~\cite{Zhang--Gladisch--Pickavet--Tao--Mohr2010,Li--Xu--Swami--Himayat--Fettweis2011,han--leung--niu--ristaniemi--seet2011,Capone--Kilper--Niu2012,Hossain--Bhargava--Fettweis2012} and the references therein for various issues in energy efficiency and management in communication systems. 

In this paper, we consider mitigating the variability of renewable energy sources through geographical diversity. We consider the case when two or more BSs are connected by power lines so as to allow for transfer of energy between each other. A transfer of energy between two BSs allow for one that has excess of energy to compensate the other that has a deficit due to either higher demand of the users connected to the BS, or lower generation of renewable energy. We analyze the reduction in conventional energy needed to power the BSs if they are allowed to transfer energy, even when there is storage inefficiency and resistive power loss. We consider the availability of different information about the renewable energy sources and demand for our setting, and propose algorithms that take advantage of the energy cooperation between BSs and the information available to minimize their energy consumption from conventional sources. 

Another motivation for considering energy transfer comes from the possibility of using the power line as a backhaul link to enable coordinated multipoint transmission (CoMP) for cellular BSs~\cite{Gesbert2010}. This results in an attractive dual use of the power line for both energy cooperation and communication cooperation. 
 
The rest of this paper is organized as follow. In Section~\ref{sect:2}, we give formal definitions and description of our proposed energy cooperation model. In Section~\ref{sect:3}, we study the optimal offline energy cooperation policy for the case of \textit{deterministic} renewable energy and demand profiles in which the future renewable energy and demand are known in advance. This setup, which has also been considered in energy harvesting based wireless communications~\cite{Ozel2011},~\cite{Ho2012}, models the scenario where we have good approximations of the renewable energy and demand profiles for the duration of interest and are willing to ignore small prediction errors. In Section~\ref{sect:4}, we consider the general case of \textit{arbitrary} renewable energy and demand profiles, and propose an online energy cooperation policy for this case. We analyze the optimality properties of this online policy under certain conditions, and compare its performance with the lower bound obtained by the offline policy via simulation. In Section~\ref{sect:hyb}, we propose an online hybrid algorithm that incorporates some offline information about the energy profile, and compare the performance of this hybrid algorithm to the online algorithm.  Finally, in Section~\ref{sect:5}, we conclude the paper and discuss some possible extensions for future work.  

\section{System Model} \label{sect:2}
In this paper, our focus will be on the case of two base stations, namely BS 1 and BS 2, with individual renewable energy generators, conventional energy sources, energy storage devices and connected with a power line. Our model, as depicted in Fig.~\ref{fig1}, can be easily generalized to multiple (more than two) BSs, but we consider only the case of two BSs in this paper for simplicity. 
\begin{figure}
\begin{center}
\psfrag{1}[c]{ }
\psfrag{2}[c]{ }
\psfrag{s}[c]{Storage}
\psfrag{bs1}[c]{BS 1}
\psfrag{bs2}[c]{BS 2}
\psfrag{b}[c]{$\beta$}
\psfrag{a}[c]{$\alpha$}
\psfrag{re1}[c]{$RE_1$}
\psfrag{re2}[c]{$RE_2$}
\psfrag{d1}[c]{$DE_1$}
\psfrag{d2}[c]{$DE_2$}
\psfrag{w2}[c]{$W_2$}
\psfrag{w1}[c]{$W_1$}
\scalebox{0.7}{\includegraphics{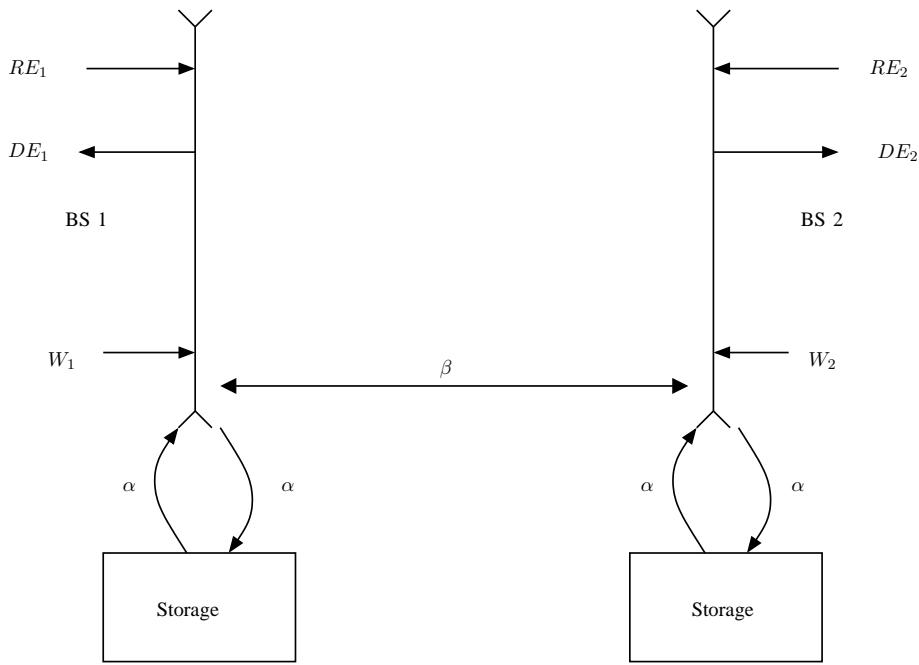}}
\caption{System setup} \label{fig1}
\end{center}
\vspace{-20pt}
\end{figure}

We consider a finite-horizon time-slotted system with slot index $t$, $1 \le t \le N$, and $N$ denoting the total number of slots under investigation. In the following, we define the elements of our energy cooperation model with two BSs, i.e. BS 1 and BS 2. We will use $i \in \{1,2\}$ to denote an element at the corresponding base station.
\subsection{Model Elements}
\textbf{Renewable energy generated at BS $i$ and time $t$}: $RE_i(t) \ge 0$.

\textbf{Demand at BS $i$ and time $t$}: $DE_i(t)\ge 0$. 

\textbf{Net energy generated at BS $i$ and time $t$}: $E_i(t) = RE_i(t) - DE_i(t)$. This quantity can be positive, representing a surplus, or negative, representing a deficit.

\textbf{Energy stored in BS $i$ at time $t$}: $s_i(t)\ge 0$. To model limited storage constraint, we further assume $s_i(t) \le S_{\rm max}$.

\textbf{Energy charged/discharged to/from storage at BS $i$ and time $t$}: $c_i(t) \ge 0$ /$d_i(t)\ge 0, d_i(t) \le s_i(t)$. Intuitively, for given BS $i$ and time $t$, there is at most one of $c_i(t)$ and $d_i(t)$ that is strictly positive, i.e. $c_i(t)\cdot d_i(t) = 0$.

\textbf{Energy transfer from BS 1 (or 2) to BS 2 (or 1)}: $x_{12}(t) \ge 0$ (or $x_{21} \ge 0$). For a given time $t$, there is at most one of $x_{12}(t)$ and $x_{21}(t)$ that is strictly positive, i.e. $x_{12}(t)\cdot x_{21}(t) = 0$.

\textbf{Energy drawn from conventional energy source at BS $i$ and time $t$}: $w_i(t) \ge 0$.
 
\subsection{System Dynamics}

We require the following equations for \textbf{storage dynamics} to be satisfied: $s_i(t+1) = s_i(t) + \alpha c_i(t) - d_i(t)$. Here, $0 \le \alpha \le 1$ represents storage inefficiency, i.e. the energy lost in storage. As discussed earlier, we also require $0 \le s_i(t) \le S_{\rm max}$ for all $t$. The combined storage dynamics and constraint leads to the constraint: $-s_{i}(t)\le \alpha c_i(t) -d_i(t) \le S_{\rm max} - s_i(t)$. We also assume that $s_1(1) = s_2(1) = 0$. That is, there is no energy in storage at the initial time\footnote{This assumption is made for simplicity of exposition and can be generalized to arbitrary storage values.}. Furthermore, the following two inequalities need to be satisfied at BS 1 and BS 2, respectively, in order to maintain their \textbf{energy neutralization} at each time $t$:
\begin{align}
& E_1(t) + w_1(t) - c_1(t) + \alpha d_1(t) - x_{12}(t) + \beta x_{21}(t) \ge 0, \label{eq1} \\ 
& E_2(t) + w_2(t) - c_2(t) + \alpha d_2(t) - x_{21}(t) + \beta x_{12}(t) \ge 0. \label{eq2}
\end{align}
Here, $\alpha$ again represents storage inefficiency and captures in this case, the inefficiency in drawing energy from storage, while $0\le \beta \le 1$ represents \textit{resistive loss} in transferring energy from one BS to another.~\eqref{eq1} captures the constraint that any demand at time $t$ at BS 1 has to be satisfied, by perhaps a combination of discharge from storage, transfer from BS 2, and conventional energy, or renewable energy. Similarly,~\eqref{eq2} captures the energy balance requirement for BS 2. 
\subsection{Control Policy and Objective Function}

In general, $E_1(t)$ and $E_2(t)$ can be modeled by a jointly distributed continuous stochastic process with a joint distribution $F$. Using vector notation, for any scalars $y_1(t), y_2(t), \ldots, y_n(t)$, we let $y(t) = [y_1(t), y_2(t), \ldots, y_n(t)]^T$. Hence, we let $s(t) = [s_1(t), s_2(t)]^T$ represent the state of our system at time $t$. Similarly, our \textit{control variables} at time $t$ are the tuples $(w(t), c(t), d(t), x_{12}(t), x_{21}(t))$, where $w(t) = [w_1(t), w_2(t)]^T, c(t) = [c_1(t), c_2(t)]^T, d(t) = [d_1(t), d_2(t)]^T$. In general, these control variables at time $t$ are functions of the past history, $\{(E(k)), 1\le k\le t\}$, with $E(k) = [E_1(k), E_2(k)]^T$, and the joint distribution $F$. A control policy $\pi$ is then a sequence of these control variables\footnote{The use of the symbol $\pi$ to represent a control policy is standard in the control/dynamic programming literature. With an abuse of notation, we will also be using the symbol $\pi$ to represent the number 3.14159... in our numerical simulations. It will be clear from context whether we are using the symbol $\pi$ for a control policy or the number.}.
That is, $\pi = \{(w(t), c(t), d(t), x_{12}(t), x_{21}(t)), 1\le t\le N\}$.

Next, the objective of our setup is to minimize the expected average conventional energy consumed. That is, we seek a control policy $\pi^*$ that minimizes
\begin{align*}
\E\left( \frac{1}{N} \sum_{t=1}^N \left(w_{1}(t) + w_2(t)\right) \right),
\end{align*}
where the expectation is taken with respect to the joint distribution $F$, and under the control policy $\pi^*$.

\begin{remark}
Another valid cost criteria is to let $N \to \infty$ and minimize the long-run expected average conventional energy cost. That is, we wish to minimize
\begin{align*}
\limsup_{N\to \infty}\E\left( \frac{1}{N} \sum_{t=1}^N \left(w_{1}(t) + w_2(t)\right) \right).
\end{align*}
This criterion has the advantage of being insensitive to the starting state, but intuition about our model can be more easily obtained when $N$ is finite. In this paper, we will restrict our attention to finite $N$ for simplicity.
\end{remark}  

The optimal control policy for our model, as currently formulated, is open. In the rest of this paper, we will consider a number of special cases in which we can obtain some useful insight on this problem.

\section{Offline Algorithm with Deterministic Energy Profile} \label{sect:3}
The first restriction that we make to this model is to consider a deterministic energy profile, with the net energy profile $E_1(t)$ and $E_2(t)$ being known to both BSs for all $t$. In this case, our model reduces to the following linear program.

\begin{theorem} \label{thm:1}
When the net energy profiles $E_1(t)$ and $E_2(t)$ are deterministic and known to BS 1 and BS 2 for all $t$, the optimal control policy, $\pi^*$, is found by solving the following linear program.{\allowdisplaybreaks
\begin{align}
& \underset{\pi}{\min} \sum_{t=1}^N (w_1(t) + w_2(t))\nonumber\\
& \text{subject to (for $1 \le t \le N$)} \nonumber\\
& s(t+1) = s(t) + \alpha c(t) - d(t), \label{eqc1}\\
& E_1(t) + w_1(t) - c_1(t) + \alpha d_1(t) - x_{12}(t)  + \beta x_{21}(t) \ge 0, \label{eqc2} \\ 
& E_2(t) + w_2(t) - c_2(t) + \alpha d_2(t) - x_{21}(t)  + \beta x_{12}(t) \ge 0, \label{eqc3}\\
& [0, 0]^T \le s(t) \le [S_{\rm max}, S_{\rm max}]^T, \label{eqc4}\\
& d(t) \le s(t), \label{eqc5}\\
& s(1) =0, c(t), d(t) \ge 0, x_{12}(t), x_{21}(t) \ge 0. \label{eqc6}
\end{align}}
\end{theorem}
\begin{proof}
The reduction to the linear program follows from the assumption that the energy profiles are known for all $t$. In this case, the objective function simply reduces to the sum of the conventional energy required at each time $t$. Note that in the above problem, we do not explicitly put the constraints $c_i(t)d_i(t) = 0$, $i = 1,2$, and $x_{12}(t)x_{21}(t) = 0$ for any given $t$. However, it can be shown that the optimal solution of this problem always satisfies these constraints, and thus there is no loss of optimality in removing such constraints. 
\end{proof}

It is easy to see that there can be several solutions achieving the same objective value in the linear program formulation in Theorem~\ref{thm:1}. In addition to minimizing the energy drawn from the grid, a secondary objective could be to maximize the sum of the energy stored in the BSs' storages at time $N+1$, so that the stored energy could be used in future time blocks to reduce the energy drawn from the grid. In this case, we can add another optimization step to maximize the energy stored in the base stations at time $N+1$, subject to the constraint that the minimum amount of energy is drawn from the main grid. This is shown in the following algorithm.

\begin{algorithm}[H]
\caption{Offline storage maximization with minimum conventional energy consumption}
\label{algo1}
\algsetup{indent=1.5em,
linenosize=\footnotesize}
\begin{algorithmic}[1]
\STATE {Input: $E_1(t)$ and $E_2(t)$ for $1\le t \le N$}
\STATE {Solve Linear program in Theorem~\ref{thm:1} and output $V_1$, the optimal value of the linear program.}
\STATE Solve the following linear program
\begin{align}
& \underset{\pi}{\max}\; s_1(N+1) + s_2(N+1) \nonumber\\ 
& \text{subject to (for $1 \le t \le N$)} \nonumber\\
& \sum_{t=1}^N (w_1(t) + w_2(t)) \le V_1, \nonumber \\
& \mbox{Equations } \eqref{eqc1} \mbox{ to } \eqref{eqc6} \nonumber
\end{align}
\STATE {Output: $\pi^*$, the optimal policy that minimizes energy consumption and maximizes storage at time $N+1$}
\end{algorithmic}
\end{algorithm}

The assumption of deterministic energy profile models the case when the demands and renewable energy generation can be well approximated for $1\le t \le N$; i.e., the case when the error in predicting the demand and renewable energy generated is small. Furthermore, it also allows us to gain insight into situations where it is beneficial for BSs to cooperate with each other. Intuitively, energy cooperation is helpful whenever the net energy generated at the two BSs are sufficiently uncorrelated or anti-correlated, as will be shown next.

To demonstrate the benefits of energy cooperation for two BSs, we model $E_1(t)$ and $E_2(t)$ with the following energy profiles.
\begin{align}
E_1(t) &= A\sin(\omega t), \label{en1}\\
E_2(t) &= A\sin(\omega t + \theta).\label{en2}
\end{align}
Here, the correlation between the net energy profiles at BSs 1 and 2 is measured by the phase shift $\theta$. This approach of modeling correlation has been used in related context, such as in work on communications with energy harvesting devices~\cite{CHuang2012}.

\textit{Energy saving versus storage for different $\theta$}: We now show some simulation results on the energy saving versus storage for different values of $\theta$. We set the following values: $\omega = 2\pi/24$, $A = 3$, $0\le t \le 239$, $\theta \in \{\pi/4, \pi/2, 3\pi/4, \pi\}$, $\beta = 0.8$ and $\alpha =0.9$. The results are plotted in Fig.~\ref{fig2}.
\begin{figure}
\begin{center}
\scalebox{0.6}{\includegraphics{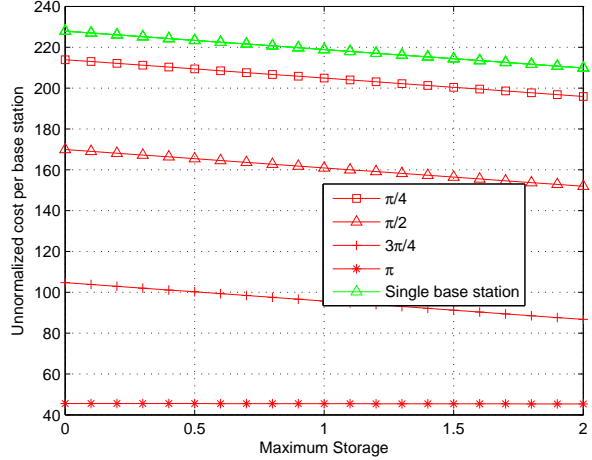}} 
\caption{Conventional energy consumed versus storage for different $\theta$} \label{fig2}
\end{center}
\vspace{-20pt}
\end{figure}
We compare the average unnormalized cost $\sum_{t=0}^{239} (w_1(t) + w_2(t))/2$ against that of a single BS having the energy profile in~\eqref{en1} (plotted in green in the figure).

As we can see from the figure, BSs' energy cooperation helps in general as the average cost per BS for the two cooperating BSs is lower than that of a single BS. As $\theta$ varies from $\pi/4$ to $\pi$, the cost per BS decreases as the energy profiles of the two BSs become more anti-correlated.

As storage increases, it is also clear that the cost decreases, since more of the excess energy generated can be stored for later use, when there is a deficit. This storage benefit, however, decreases with $\theta$ increasing to $\pi$. Increasing $\theta$ to $\pi$ signifies an increase in geographical diversity, resulting in the ability to compensate for deficit at one BS with excess from the other BS. When $\theta = \pi$, there is little benefit from increasing storage. 

\textit{Energy savings versus $\theta$ for fixed storage}: To show the effect of $\theta$ more clearly, we now keep the storage fixed at $S_{\rm max} =1$ and vary $\theta$ from $0$ to $2\pi$. The rest of the parameters are kept fixed. In Fig.~\ref{fig3}, we plot the percentage cost savings, relative to the energy cost of a single BS with the energy profile of~\eqref{en1}, against different values of $\theta$.
\begin{figure}
\begin{center}
\includegraphics[width = 0.525\textwidth]{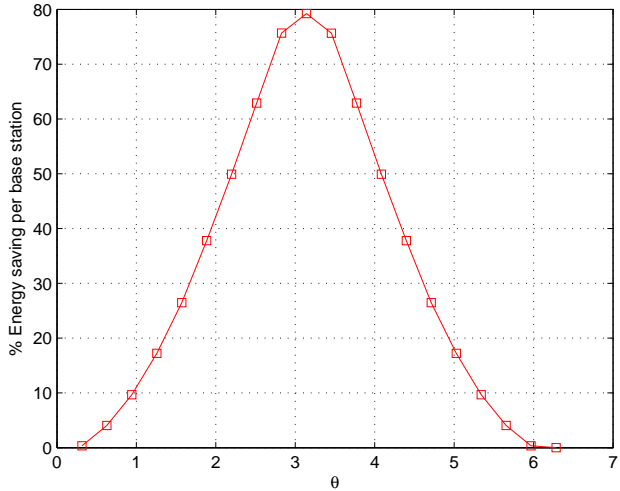}
\caption{Percentage energy saving versus $\theta$ for $S_{\rm max} =1$} \label{fig3}
\vspace{-20pt}
\end{center}
\end{figure}
As we can see from the figure, the saved cost increases as $\theta$ varies from $0$ to $\pi$, at which point the energy profile of BS 2 is anti-correlated with BS 1. This allows effective compensation through energy transfer between the two BSs. As $\theta$ varies from $\pi$ to $2\pi$, the energy profile becomes highly correlated again, resulting in fewer opportunities to perform energy transfer between the two BSs.

\section{Online Algorithm with Stochastic Energy Profile} \label{sect:4}
We now consider the more practical case when the net energy at both BSs are stochastic and not known ahead of time. We propose an online energy cooperation algorithm based on a greedy heuristic for minimizing conventional energy usage in Section~\ref{subsect:3a}. We then analyze some properties of this algorithm in Section~\ref{subsect:3b}. In particular, we state some optimality properties under specific energy profiles. Finally, in Section~\ref{subsect:3c}, we provide simulation results on the performance comparison between the online algorithm versus the optimal offline algorithm proposed in Section~\ref{sect:3}. 

To describe the algorithm, we first assume $\alpha >0$ and $\beta >0$ to avoid the complications of dealing with the case of no storage ($\alpha =0$) or no cooperation between BSs ($\beta =0$).

\subsection{Greedy Online Algorithm} \label{subsect:3a}
Our greedy algorithm follows from considering a single snapshot of the linear programs given in the previous section with arbitrary storage states. That is, with $N = 1$, but with the additional condition that the initial storage states need not be equal to zero. We now present our algorithm as follow.

\textit{Greedy Algorithm step 1: Greedy minimization of energy drawn from the main grid. }
Assume that the initial storage values are $s_1$ for BS 1 and $s_2$ for BS 2. Then, we solve the following linear program (for notational simplicity, we suppress the dependence on time){\allowdisplaybreaks
\begin{align}
& {\min}\; (w_1 + w_2)\nonumber \\
& \text{subject to} \nonumber\\
& [0, 0]^T \le [s_1, s_2]^T + \alpha c - d \le [S_{\rm max}, S_{\rm max}]^T, \label{eqg1}\\
& E_1 + w_1 - c_1 + \alpha d_1 - x_{12} + \beta x_{21} \ge 0,  \\ 
& E_2 + w_2 - c_2 + \alpha d_2 - x_{21} + \beta x_{12} \ge 0, \\
& d \le [s_1, s_2]^T, \\
& c\ge 0, d \ge 0,  x_{12}\ge 0, x_{21} \ge 0. \label{eqg5}
\end{align}}

\textit{Greedy Algorithm step 2: Storage maximization. } Let $V_1$ be the optimal value of the linear program in step 1. Then, we solve the following linear program.{\allowdisplaybreaks
\begin{align*}
& {\max}\; [s_1, s_2]^T + \alpha c - d\\
& \text{subject to} \\
& w_1 + w_2 \le V_1, \\
& \mbox{Equations } \eqref{eqg1} \mbox{ to } \eqref{eqg5}.
\end{align*}}

In the case of a single snapshot, instead of solving two linear programs individually, we can combine the linear programs, as stated in the next proposition.
\begin{proposition} \label{prop:greed_lin}
Let $\mathbf{1}$ be the vector of all ones. For any $\gamma$ with $0< \gamma < \alpha\beta$, the greedy algorithm is obtained by solving the following linear program. {\allowdisplaybreaks
\begin{align*}
& \min\; (w_1 + w_2) - \gamma \mathbf{1}^T([s_1, s_2]^T+ \alpha c -d)\\
& \text{subject to} \\
& \mbox{Equations } \eqref{eqg1} \mbox{ to } \eqref{eqg5}.
\end{align*}}
\end{proposition}
We defer the proof of this Proposition to Appendix~\ref{appen:1a}.

The greedy algorithm requires solving a small linear program. For the case of two BSs, however, it can be further simplified to the following equivalent algorithm in Proposition~\ref{prop:algo} by considering the actions the BSs would take at each time $t$. 

\begin{proposition} \label{prop:algo}
The greedy online linear program is equivalent to the following algorithm, which should be understood to be implemented for each time $1 \le t \le N$, and we again suppress the dependence on $t$ for convenience. Unless otherwise stated, we set all of $w, c, d, x_{12}$ and $x_{21}$ equal to zeros in the algorithm. For each time $t$, if

\textit{Case 1}: $E_1\ge 0$ and $E_2 \ge 0$. For $i \in \{1,2\}$, we first carry out the following
\begin{align*}
c_{i} &= \min\{(S_{\rm max} - s_i)/\alpha, E_i\}, \\
s_i &\leftarrow s_i + \alpha c_i.
\end{align*}
If both $s_1 = s_2 = S_{\rm max}$ or $s_1, s_2 < S_{\rm max}$, this case terminates. Otherwise, if $s_2 < S_{\rm max}$ and $s_1 = S_{\rm max}$, BS 1 transfers energy to BS 2 for storage. That is, we set
\begin{align*}
x_{12} &= E_1 - c_1, \\
c_2' &= \min\{ \beta x_{12}, (S_{\rm max} - s_2)/\alpha\}, \\
s_2 &\leftarrow s_2 + \alpha c_2',\\
c_2 &\leftarrow c_2 + c_2'.
\end{align*}

Similarly, if $s_1 < S_{\rm max}$ and $s_2 = S_{\rm max}$, the roles of BSs 1 and 2 in the above are reversed.\\

\textit{Case 2A}: $E_1 \ge 0$, $E_2 < 0$ and $\beta \ge \alpha^2$. In this case, BS 1 first transfers the net energy to BS 2 to compensate for the deficit. Hence, let 
\begin{align*}
x_{12} &= \min\{E_1, |E_2|/\beta\}, \\
E_2' &= E_2 + \beta x_{12}.
\end{align*}
Now, if $E_2' = 0$, we carry out the algorithm in Case 1 with net energy profiles $E_1' = E_1 -x_{12}$ and $E_2'$. If $E_2' <0$, we compensate for the remaining deficit via storage at BS 2 first. We set
\begin{align*}
d_2 &= \min\{|E_2'|/\alpha, s_2\}, \\
E_2'' &= E_2' + \alpha d_2.
\end{align*}
If $E_2'' = 0$, this case is completed. Otherwise, we compensate from storage at BS 1. That is, we set
\begin{align*}
d_1 &= \min\{|E_2''|/(\alpha\beta), s_1\}, \\
x_{12}  &\leftarrow x_{12} +  \alpha d_1,\\
E_2''' &= E_2'' + \alpha \beta d_{1}.
\end{align*}
Finally, if there is still a deficit remaining ($E'''_2 <0 $), we compensate through conventional energy consumption and set $w_2 = |E'''_2|$.

\textit{Case 2B}: $E_1 \ge 0$, $E_2 < 0$ and $\beta < \alpha^2$. In this case, BS 1 tries to maximize its own storage level using the excess energy, subject to minimizing energy required to compensate for the deficit at BS 2. The optimal policy is determined under the following two sub-scenarios.
\begin{itemize}
\item $|E_2| \ge \beta E_1 + \alpha s_2$: The optimal policy is the same as Case 2B.
\item $|E_2| < \beta E_1 + \alpha s_2$: The optimal policy is given as follows. Let $1_{(.)}$ be the indicator function.
\begin{align*}
x_{12} & = \max\left\{\frac{|E_2| - \alpha s_2}{\beta}, E_1 - \frac{S_{\rm max} - s_1}{\alpha}, 0\right\}, \\
c_1 &= E_1 - x_{12}, \\
d_2 & = \max\left\{\frac{|E_2| - \beta x_{12}}{\alpha}, 0\right\}, \\
c_2 & = \mathbf{1}_{d_2 = 0} \min\left\{\frac{S_{\rm max} - s_2}{\alpha}, \beta x_{12} - |E_2|\right\}, \\
s_1 &\leftarrow s_1 + \alpha c_1, \\
s_2 &\leftarrow s_2 + \alpha c_2 -d_2. 
\end{align*}
\end{itemize}

\textit{Case 3}:  $E_1 < 0$, $E_2 \ge 0$ and $\beta \ge \alpha^2$. This case is symmetric to Case 2, with the roles of BSs 1 and 2 reversed. We therefore omit the description of the algorithm here.

\textit{Case 4}:  $E_1 < 0$ and $E_2 < 0$. In this case, each BS compensates using individual storage first, before helping the other. That is, for $i\in\{1,2\}$, we let 
\begin{align*}
d_i &= \min\{s_i, |E_i|/\alpha\}, \\
E'_i &= E_i + \alpha d_i, \\
s_i &\leftarrow s_i - d_i.
\end{align*}
If either $E'_1 \ge 0$ or $E'_2 \ge 0$, the algorithm reduces to the first three cases with net energy profiles being $E'_1$ and $E'_2$. If both $E'_1<0$ and $E'_2<0$, we compensate with conventional energy generation and set $w_i = |E'_i|$. 

\end{proposition}


Proof of this Proposition is deferred to Appendix~\ref{appen:1}. For the case of no storage ($\alpha =0$), the greedy algorithm is modified in the obvious manner by not sending any energy to storage or drawing energy from storage. It is easy to show that this modified greedy algorithm is optimal for arbitrary energy profiles.

For the case of no cooperation ($\beta =0$), the greedy algorithm is again modified in the obvious manner by not requiring transfers between two BSs. Optimality of the greedy algorithm for this case will be discussed in the next subsection. 

\subsection{Optimality Properties} \label{subsect:3b}
Although the greedy algorithm is a conceptually simple one, it has several optimality properties that we now analyze. To keep the discussion clear and provide intuition on this policy, several proofs of the results are deferred to the Appendices. We will not suppress the dependence on $t$ in this subsection as we will consider the energy profiles over time. 

\begin{proposition}
If $\beta =0$ or $\beta =1$, the greedy algorithm is optimal for arbitrary energy profiles.
\end{proposition}
\begin{proof}
When $\beta =1$, the system reduces to a single BS with $E(t) = E_1(t) + E_2(t)$ and $0 \le s(t) \le 2S_{\rm max}$. The optimality of the greedy algorithm can then be inferred from~\cite[Theorem 1]{Su2011}.

For the case of $\beta =0$, no cooperation between the two BSs is possible. Optimality of the modified greedy algorithm for this case follows again from the fact that the greedy algorithm is optimal for individual BSs~\cite[Theorem 1]{Su2011}.
\end{proof}

We now proceed to analyze the greedy algorithm for the non-boundary cases of $0< \beta < 1$. It will be useful to define the following quantities. Define the unnormalized \textit{cost-to-go} function under policy $\pi$ at time $t$ and state $s(t)$ as 
\begin{align}
J_{\pi}(s(k)) = \E \left(\sum_{t = k}^N \mathbf{1}^Tw_{\pi}(t)\right).
\end{align} 
We also denote the optimal cost-to-go function under the optimal policy $\pi^*$ as $J_{\pi^*}(s(t))$. For our setting, $J_{\pi}(s(t))$ has a number of useful properties, which we now state. The results that follow in the rest of this section hold for any stochastic net energy profiles. Therefore, in our proofs, we will suppress the dependence of the control policy on the joint distribution of $E_1(t)$ and $E_2(t)$.
\begin{proposition} \label{prop2}
Suppose $s'(k) \ge s(k)$ component-wise. Then,
\begin{align}
J_{\pi^*}(s(k)) \le J_{\pi^*}(s'(k)) + \alpha \mathbf{1}^T(s'(k) - s(k)). \label{eq:1}
\end{align}
\end{proposition}
The proof is given in Appendix~\ref{append:prop2}
It is obvious that a system starting with higher stored energy states has a lower optimum cost. Proposition~\ref{prop2} quantifies the maximum additional cost incurred by a system starting from a lower storage state. As an example, Proposition~\ref{prop2} formalizes the obvious fact that energy drawn from the main grid is never used to increase storage levels. 
\begin{corollary} \label{coro1}
For an optimal policy, energy is never drawn from the grid to increase the storage levels. That is, for any $\Delta_1, \Delta_2, \Delta_W \ge 0$ such that $\Delta_W = \Delta_1 + \Delta_2$, we have
\begin{align*}
J_{\pi^*}\left(\left[\begin{array}{c}s_1(k) + \alpha \Delta_1 \\ s_2(k)+ \alpha \Delta_2\end{array}\right]\right) + \Delta_W & \ge J_{\pi^*}\left(\left[\begin{array}{c}s_1(k) \\ s_2(k)\end{array}\right]\right),
\end{align*} 
which implies that the optimal energy to be drawn, $\Delta_W$, should be zero.
\end{corollary}
Proof of this corollary is immediate from Proposition~\ref{prop2}. 

The bound given in Proposition~\ref{prop2} can be strengthened in various ways if more information about the energy profiles are known. We state the following claim that will be used in the sequel.
\begin{claim} \label{clm:bound}
If $E(t) \ge 0$ for $t \ge k$ and $\Delta \le (S_{\rm max} - s_1(k))/\alpha$,
\begin{align*}
J_{\pi^*}\left(\left[\begin{array}{c}s_1(k) \\ s_2(k)\end{array}\right]\right) & \le J_{\pi^*}\left(\left[\begin{array}{c}s_1(k) + \Delta\\ s_2(k)\end{array}\right]\right) + \alpha \beta \Delta.
\end{align*}
\end{claim}
The proof of this claim follows similar arguments to that in Proposition~\ref{prop2} and uses the assumption that $E_1(t) \ge 0$ for all $t \ge k$. As the proof is quite similar to that in Proposition~\ref{prop2}, we will omit the proof here. Instead, we give the intuition for this claim. Since $E_1(t) \ge 0$, the additional stored energy can only be used to compensate for any deficit at BS 2. The total additional energy that can be sent to BS 2 is $\alpha \beta \Delta$. Instead of using storage at BS 1, we can compensate using conventional energy, incurring an additional cost of at most $\alpha \beta \Delta$. 

The same proof strategy whereby a system at a different storage state $s$ \textit{mimics} the optimal policy of the same system at storage state $s'$ can be also used to prove the following two intuitively obvious propositions. Due to space limitation, proofs of Propositions~\ref{prop3} and~\ref{prop4} are omitted in this paper. 
\begin{proposition} \label{prop3}
Case 1: It is optimal to store excess energy at each of BS 1 and BS 2 first if there is still storage available, rather than to transfer the energy between them for storage. More concretely, suppose $\Delta >0$ units of energy is available at BS 1 at time $t = k$ with $s_1(k) + \alpha \Delta \le S_{\rm max}$ and $s_2(k) + \alpha \beta \Delta \le S_{\rm max}$, then
\begin{align*}
J_{\pi^*}\left(\left[\begin{array}{c}s_1(k) + \alpha \Delta \\ s_2(k)\end{array}\right]\right) \le J_{\pi^*}\left(\left[\begin{array}{c}s_1(k) \\ s_2(k) + \alpha \beta \Delta\end{array}\right]\right).
\end{align*}

Case 2: Suppose that there is a deficit of $-\Delta$ ($\Delta \ge 0$) at BS 1 such that $\Delta \le \min\{\alpha s_1(k), \alpha\beta s_2(k)\}$, which has to be compensated by either storage at BS 1 or storage at BS 2. Then, it is optimal to compensate using storage at BS 1 rather than storage at BS 2. That is, 
\begin{align*}
J_{\pi^*}\left(\left[\begin{array}{c}s_1(k) - \Delta/\alpha \\ s_2(k)\end{array}\right]\right) \le J_{\pi^*}\left(\left[\begin{array}{c}s_1(k) \\ s_2(k) -  \Delta/(\alpha \beta)\end{array}\right]\right).
\end{align*}

\end{proposition}
Proof of this Proposition is deferred to Appendix~\ref{pf:storage}
\begin{proposition} \label{prop4} If $\beta > \alpha$, then energy transfer is always optimal. That is, if at time $t = k$, $E_1(k) >0 >  E_2(k)$, we can assume without loss of generality that sending $\Delta = \min\{|E_2(k)|/\beta, E_1(k)\}$ units of energy from BS 1 to BS 2 at time $t = k$ to compensate for the deficit of $E_2(k)$ is part of an optimal policy. More formally, let $J_{\pi^*}(s(k), E_1(k), E_2(k))$ be the optimal cost to go function\footnote{We suppress the dependence of $\pi^*$ on past histories here.}. Then,
\begin{align*}
J_{\pi^*}(s(k), E_1(k)- \Delta, E_2(k)+ \Delta) &\le J_{\pi^*}(s(k), E_1(k), E_2(k)). 
\end{align*}
\end{proposition}
Proof of this Proposition is deferred to Appendix~\ref{pf:transfer}.

\textit{Remarks on Propositions~\ref{prop3} and~\ref{prop4}}: As with Proposition~\ref{prop2}, the above two propositions show formally some intuitive aspects of energy cooperation that we would expect in such a system, and also certain optimality aspects of the greedy algorithm. Proposition~\ref{prop3} shows the intuitively obvious fact that it is better to store energy locally rather than to store the energy at storage of the other BS. Proposition~\ref{prop4} formalizes the notion that if $\beta> \alpha$, then it is more cost efficient to transfer energy to help the other BS rather than to store energy for future use, since the proportional loss in energy storage ($1- \alpha$) is higher than that in energy transmission ($1 - \beta$). Observe that the condition $\beta > \alpha$ is a special case of Case 2A of our greedy algorithm, in which energy transfer is always carried out first, rather than compensating using local storage. 

Next, using Propositions~\ref{prop2} to~\ref{prop4}, we arrive at the optimality of the greedy algorithm for some special cases of the energy profiles. 

\begin{proposition} \label{prop5}
If $E_1(t) \ge 0$ for all $1 \le t \le N$ and $\beta > \alpha$, then the greedy policy is optimal. By symmetry, the same result holds if, instead, $E_2(t) \ge 0$ for all $1 \le t \le N$ and $\beta > \alpha$.
\end{proposition}
Proof of this Proposition is given in Appendix~\ref{append:prop5}. The intuition behind Proposition~\ref{prop5} is that if $E_1(t)\ge 0$ for all $t$, then energy should be transferred to BS 2 to compensate for any possible deficit. The condition of $\beta > \alpha$ ensures that it is always more efficient to transfer energy to BS 2 in the current time step than to store it for possible use in later time steps. 

The condition that $\beta > \alpha$ can be relaxed, if more assumptions can be made about the energy profile $E_2(t)$.
\begin{proposition} \label{prop6}
If $E_1(t) \ge 0$ and $E_2(t)\le 0$ for all $1 \le t \le N$, then the \textit{modified} greedy policy, in which the policy in Case 2A of Proposition~\ref{prop:algo} is implemented at each time $t$, is an optimal policy. Similar to Proposition~\ref{prop5}, by symmetry, the same result holds if $E_2(t) \ge 0$ and $E_1(t)\le 0$. 
\end{proposition}
Proof of this Proposition is given in Appendix~\ref{append:prop6}. Proposition~\ref{prop6} essentially shows that energy transfer is always optimal if one BS is always in surplus while the other BS is always in deficit. The intuition behind this Proposition is the following observation. When $\beta \le \alpha$, it can be more efficient in general to compensate for any deficit from local storage at BS 2 than to transfer energy from BS 1, since storage at BS 2 can be recharged more efficiently in future time steps with $\beta \le \alpha$. However, when $E_2(t)$ is always non-positive, any charging of storage at BS 2 has to come from BS 1, which incurs a proportional loss of $1-\alpha \beta\ge 1-\alpha$, resulting in any discharging or charging being less efficient than energy transfer. 

\subsection{Numerical Results} \label{subsect:3c}
We now compare the greedy online algorithm to the optimal offline algorithm proposed in Section~\ref{sect:3}. Due to the lack of future information on the energy profiles, the greedy algorithm clearly cannot do as well as the offline algorithm, except under conditions discussed in Section~\ref{subsect:3b}. We compare the performance differences between the two algorithms when the conditions in Section~\ref{subsect:3b} are not satisfied. We adopt the same simulation setting as in Section~\ref{sect:3}.
\begin{figure}
\begin{center}
\includegraphics[width = 0.525\textwidth]{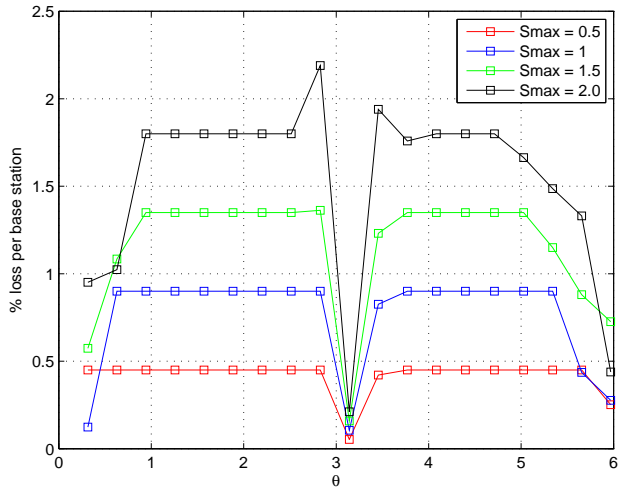}
\caption{Percentage energy loss of greedy online algorithm w.r.t. offline algorithm versus $\theta$} \label{fig0}
\vspace{-20pt}
\end{center}
\end{figure}

Fig.~\ref{fig0} shows the performance gap between the greedy online algorithm and the optimal offline algorithm that has access to the entire energy profile. Somewhat surprisingly, the greedy online algorithm suffers only a small loss (maximum of 2.3\%) compared to the offline algorithm under the sinusoidal energy profiles assumed in~\eqref{en1} and~\eqref{en2}. Note that the vertical axis shows the loss of the greedy online algorithm, which is the percentage increase in energy consumption with respect to the optimal offline algorithm over the \textit{entire} time horizon of $N$. Furthermore, when the energy profiles at the two BSs are anti-correlated ($\theta = \pi$), or highly positively correlated ($\theta$ small or close to $2\pi$), the percentage loss due to using an online algorithm is significantly lower, and the greedy online algorithm is almost as efficient as the offline algorithm. The offline algorithm, however, seems to be able to make more intelligence use of storage, resulting in generally higher energy saving when the the maximum amount of storage is increased. This suggests that greedy charge and discharge strategy may not be optimal in general, as can be expected.

\section{Hybrid model and algorithm} \label{sect:hyb}
In the previous two sections, we consider two cases for energy cooperation that can be thought of as being in two different extremes. In the first case, we assume that the energy profile is known entirely for the duration of interest and have shown that the optimal policy can be found through solving two linear programs. In the second scenario, we assume that we do not have any statistical information about the energy profile in future time steps. By specializing the linear programs in the first case to a single time step, we proposed a greedy online algorithm and analyzed the optimality properties of the greedy algorithm under certain conditions. 

Other than these two extremes, a more realistic scenario may be to assume that some, but not complete information, is known about the future energy profiles at the BSs. For example, it may be reasonable to assume that the energy profiles consist of a deterministic waveform in which small amount of random noise is added at each time step to model the prediction errors. In such a scenario, we can combine the proposed offline and online algorithms to arrive at a \textit{hybrid} algorithm to leverage on the available information about the energy profile. Essentially, we can use the offline algorithm to determine the policy for the deterministic portion of the energy profile and then use the greedy algorithm to compensate for any differences induced by the random noise. 

Most of the definitions in our model remain unchanged. The only change in our model is the information known at the BSs with regards to the energy profiles $E_1(t)$ and $E_2(t)$. More concretely, we assume that 
\begin{align*}
E_1(t) &= E_{1d}(t) + E_{1r}(t),\\
E_2(t) &= E_{2d}(t) + E_{2r}(t),
\end{align*}
with $E_{1d}(t)$ and $E_{2d}(t)$ known to the BSs for all $t$ and at time $t = k$, the BSs know $E_{1}(t)$ and $E_{2}(t)$ for $1 \le t \le k$. The proposed algorithm is presented in Algorithm~\ref{alg1}. 
\begin{algorithm}
\caption{Hybrid Algorithm for energy minimization}
\label{alg1}
\algsetup{indent=1.5em,
linenosize=\footnotesize}
\begin{algorithmic}[1]
\STATE {Input: $E_{1d}(t)$, $E_{2d}(t)$ for $1\le t \le N$, $E_{1}(t)$ and $E_{2}(t)$ for $1\le t \le k$}
\STATE {Solve linear program proposed in Algorithm~\ref{algo1} with $E_{1d}(t)$, $E_{2d}(t)$, $1 \le t \le N$ as inputs}
\STATE {Output: $\pi_{d}$, the optimal offline policy that minimizes energy consumption and maximizes storage at time $N$}
\STATE {From $\pi_d$, we obtain the storage states $s_{1d}(t)$ and $s_{2d}(t)$ for $1 \le t \le N$}
\FOR {$t = 1 \to N$}
\STATE {$S_{1g,max}(t) = S_{\rm max} - s_{1d}(t)$}
\STATE {$S_{2g,max}(t) = S_{\rm max} - s_{2d}(t)$}
\STATE {$E_{1g}(t) = E_1(t) - w_{1d}(t) + c_{1d}(t) - \alpha d_{1d}(t) - \beta x_{21,d}(t) + x_{12,d}(t)$}
\STATE {$E_{2g}(t) = E_2(t) - w_{2d}(t) + c_{2d}(t) - \alpha d_{2d}(t) - \beta x_{12,d}(t) + x_{21,d}(t)$}
\ENDFOR
\FOR {$t = 1 \to N$}
\STATE {Solve online energy minimization problem using Proposition~\ref{prop:algo} with $E_{1g}(t)$ and $E_{2g}(t)$ as inputs and $S_{1g,max}(t)$ and $S_{2g,max}(t)$ in place of $S_{\rm max}$ as maximum storage states for BS 1 and 2, respectively.}
\STATE {Output $\pi_g(t)$, the online greedy policy at time $t$}
\STATE {Output: Hybrid policy for time $t$, $\pi(t)$, which is given by, for $i \in \{1,2\}$,
\begin{align*}
w_{i}(t) &= w_{id}(t) + w_{ig}(t), \\
c_{i}(t) &= c_{id}(t) + c_{ig}(t), \\
d_{i}(t) &= d_{id}(t) + d_{ig}(t), \\
x_{12}(t) & = x_{12,d}(t) + x_{12,g}(t), \\
x_{21}(t) & = x_{21,d}(t) + x_{21,g}(t), \\
s_{i}(t) &= s_{id}(t) + s_{ig}(t).
\end{align*}
}
\ENDFOR
\end{algorithmic}
\end{algorithm}

This algorithm may be thought of as a superposition of the offline (Algorithm~\ref{algo1}) and online (Proposition~\ref{prop:algo}) algorithms. We first solve the offline algorithm using $E_{1d}(t)$ and $E_{2d}(t)$ for all $t$. We then compensate for the part of the energy profile that we do not know ($E_{1g}(t)$ and $E_{2g}(t)$) using the online algorithm. It should be noted that the storage available for the online algorithm has to be adjusted based on the amount of storage used in the offline algorithm, and this is done by defining a variable maximum stored energy, $S_{1g,max}(t)$ and $S_{2g,max}(t)$, at each time $t$ for each of the storages. This partitioning of storage for offline and online algorithms ensure that we can separate the offline and online energy minimization problems and then combine them back again to obtain the hybrid algorithm.

We use the following energy profiles for our numerical simulation. 
\begin{align*}
E_{1}(t) &= 5sin(wt) + 0.125 E_{1r}(t), \\
E_{2}(t) &= 5sin(wt + \theta) + 0.125 E_{2r}(t), \\
S_{\rm max} &= 3.5.
\end{align*}
The simulation period $t$ and $w$ used remain the same as in previous simulations. $E_{1r}(t)$ and $E_{2r}(t)$ are independent, identically distributed Gaussian random variables with zero mean and unit variance, denoted by $N(0,1)$ for all $t$. Figure~\ref{fighybrid} plots the percentage loss of two algorithms with respect to the optimal offline algorithm that has full knowledge of $E_1(t)$ and $E_2(t)$ for all $t$, versus $\theta$. The two algorithms are the greedy online algorithm and the hybrid algorithm.
\begin{figure}
\begin{center}
\includegraphics[width = 0.525\textwidth]{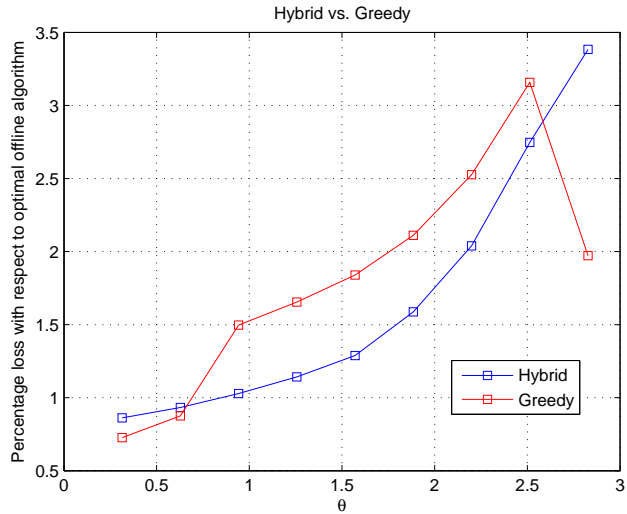}
\caption{Greedy online algorithm vs. Hybrid algorithm} \label{fighybrid}
\vspace{-20pt}
\end{center}
\end{figure}

As shown in Fig.~\ref{fighybrid}, the hybrid algorithm can outperform the greedy online algorithm for moderate values of $\theta$. This is to be expected, since the hybrid algorithm makes use of the knowledge of $E_{1d}(t)$ and $E_{2d}(t)$. When the $\theta$ is close to $0$ or $\pi$, the performance of the greedy online algorithm suffers little loss compared to the optimal offline algorithm, and the saving from knowing future values of $E_{1d}(t)$ and $E_{2d}(t)$ diminishes. This results in the greedy online algorithm being able to outperform the hybrid algorithm. 
\section{Conclusion and future directions}~\label{sect:5}
In this paper, we have proposed a model for energy cooperation between two cellular BSs with hybrid conventional and renewable energy sources, limited storages, and a connecting power line.  We consider two extreme cases. In the first case, we assume that the energy profile is known entirely for the duration of interest and have shown that the optimal policy can be found through solving a linear program. In the second scenario, we assume that we do not have any statistical information about the energy profile in future time steps. In this case, we have proposed a greedy online algorithm and analyzed the optimality properties of the greedy algorithm under some conditions. Numerical simulations comparing the offline and online algorithms were also carried out. In addition to these two extremes, another scenario is to assume that some, but not complete information, is known about the future energy profiles at the BSs. When the energy profiles consist of a deterministic waveform in which small amount of random noise is added at each time step to model the prediction errors, we proposed a hybrid algorithm that leveraged on the available information about the energy profile, and can be operated online. We compared the performance of the hybrid algorithm to the greedy online algorithm via simulations. The hybrid algorithm can outperform the online algorithm in some regimes by leveraging on the available (offline) information about the energy profiles.

Our model, while conceptually simple, can be extended in several different directions. We presented the model for two BSs in this paper, but we can readily extend the model and algorithms presented in this paper to multiple BSs. Another interesting extension would be to include pricing information into the model. A grid operator could charge different prices for conventional energy at different times of the day, and it is not difficult to extend our model to capture this pricing information. The algorithms, however, would need to change to incorporate the pricing information. 

\appendices
\section{Proof of Proposition~\ref{prop:greed_lin}} \label{appen:1a}
\begin{proof}
We will prove this proposition by contradiction. Let $w_1^*$, $w_2^*$, $s_1^* = s_1 + \alpha c_1^* -d^*_1$ and $s^*_2=s_2 + \alpha c_2^* -d^*_2$ be the optimal values found by the original greedy algorithm. Let $w'_1$, $w'_2$, $s'_1$ and $s'_2$ be the optimal values found through solving the linear program in Proposition~\ref{prop:greed_lin}. Let $V_2 = w'_1 + w'_2 \ge V_1$ (since $V_1$ is the minimal energy possible). If $V_1 = V_2$, then it is clear that $s^*_1+ s^*_2 = s'_1 + s'_2$. Hence, we need only to consider the case when $V_2 > V_1$. We have the following two cases.

\textit{Case 1:} $w_1^* \ge w_1'$ and $w_2^* \ge w_2'$. In this case, observe that the excess energy $V_2 - V_1$ can at most be used to increase the storage levels ($s_1^* + s_2^*$) by $(V_2- V_1)/\alpha\beta$, through removing the need for the storage of one base station to discharge to compensate for a deficit at the other base station. Hence, $s_1' + s_2' \le (s_1^* + s_2^*) + (V_2 - V_1)/\alpha\beta$ and
\begin{align*}
(w'_1 + w'_2) - \gamma (s'_1 + s'_2) &\ge (w^*_1 + w^*_2) - \gamma (s^*_1 + s^*_2) + (V_2 - V_1)(1 - \frac{\gamma}{\alpha\beta}) \\
&> (w^*_1 + w^*_2) - \gamma (s^*_1 + s^*_2).
\end{align*}
The last line follows from $\gamma < \alpha\beta$ and $(V_2 - V_1)>0$. Since  $w_1^*$, $w_2^*$, $s_1^*$ and $s^*_2$ are feasible for the linear program in the Proposition, this inequality contradicts the assumption that $w'_1$, $w'_2$, $s'_1$ and $s'_2$ are optimal. 

\textit{Case 2:} $w_1^* > w'_1$ and $w_2^* < w'_2$, with $w'_2 - w_2^* > w_1^* - w'_1$ since $V_2 > V_1$. In this case, we first note the following observations.
\begin{itemize}
\item Before any energy transfer, there is a deficit of $w_1^* - w'_1$ at BS 1. Otherwise, we can reduce $w_1^*$ without incurring any deficit and maintain the same $w_2^*$, which contradicts the fact that $V_1 = w_1^* + w_2^*$ is the minimum energy required from the main grid. 
\item Energy drawn from the main grid at BS 2 is never used to compensate for the deficit at BS 1. This is because one can achieve a smaller $w_1 + w_2$ by simply drawing the required energy from BS 1 and not incur the transfer cost of $(1-\beta)$ that occurs when energy is drawn from BS 2 and transferred to BS 1. 
\item The energy required to compensate for the deficit $w_1^* - w'_1$ can only come from the storage of BS 2. If storage of BS 1 is used to compensate for part of the deficit, then it means that we can achieve a $w_1 < w_1^*$ by using storage while maintaining the same $w_2^*$ (from the previous observation, none of the excess energy $w'_2 - w_2^*$ is used to compensate for the deficit at BS 1). This contradicts the fact that $w_1^* + w_2^*$ is the minimum energy required from the main grid.
\item Hence, the storage level at BS 2 must fall by $(w^*_1 - w_1')/\alpha\beta$ from $s_2^*$. 
\item On the other hand, using the same arguments as in case 1, the excess energy at BS 2 can at most lead to an increase of $(w'_2 - w_2^*)/\alpha\beta$ for the storage level at BS 1. 
\end{itemize} 

We therefore have the same inequality as case 1: $(w'_1 + w'_2) - \gamma (s'_1 + s'_2) > (w^*_1 + w^*_2) - \gamma (s^*_1 + s^*_2)$. The rest of the proof follows the same arguments as case 1. 
\end{proof}

\section{Proof of Proposition~\ref{prop:algo}} \label{appen:1}
\begin{proof}
\textit{Case 1}: $E_1\ge 0$ and $E_2 \ge 0$. In this case, clearly, $w = 0$ and both BSs try to store as much of the net energy as possible. That is, for $i \in \{1,2\}$, we first carry out the following
\begin{align*}
c_{i} &= \min\{(S_{\rm max} - s_i)/\alpha, E_i\}, \\
s_i &\leftarrow s_i + \alpha c_i.
\end{align*}
If both $s_1 = s_2 = S_{\rm max}$ or $s_1, s_2 < S_{\rm max}$, this case terminates. Otherwise, if $s_2 < S_{\rm max}$ and $s_1 = S_{\rm max}$, BS 1 transfers energy to BS 2 for storage. That is, we set
\begin{align*}
x_{12} &= E_1 - c_1, \\
c_2' &= \min\{ \beta x_{12}, (S_{\rm max} - s_2)/\alpha\}, \\
s_2 &\leftarrow s_2 + \alpha c_2',\\
c_2 &\leftarrow c_2 + c_2'.
\end{align*}

Similarly, if if $s_1 < S_{\rm max}$ and $s_2(t) = S_{\rm max}$, the roles of BSs 1 and 2 in the above are reversed.\\

\textit{Case 2}:  $E_1 \ge 0$ and $E_2 < 0$. This is the more complicated case that needs to be split into four sub-cases. 

\textit{Case 2.1}: $|E_2| \ge \beta E_1 + \alpha s_2$. In this case, all the energy is transferred to overcome the deficit. Here, we set $x_{12}' = E_1$ and $d_2 = s_2$. Set $E_2' = E_2 + \beta x_{12}' + \alpha d_2$. Next, set $d_1 = \min\{s_1, |E_2'|/\alpha\beta\}$ and $x_{12} = \alpha d_1 + E_1$. Set $w_2 = E_2+ \beta x_{12} + \alpha d_2$. \\

\textit{Case 2.2}: $\beta E_1 \le |E_2| < \beta E_1 + \alpha s_2$. In this case, $w_2 = 0$ since all the deficit can be compensated for by energy transfer and storage. The deficit is compensated for by a mixture of energy transfer and storage charge and discharge that maximizes the storage levels. We first note in this case the following simple claim.
\begin{claim} \label{clm:et}
If $\beta x_{12} \le |E_2|$ and $w_2 = 0$, then no charging occurs at BS 2; i.e. $c_2 = 0$.
\end{claim}
\begin{proof}
Suppose for the sake of contradiction that $\delta$ units of the transferred energy ($\beta x_{12}$) is used to charge the storage at BS 2 instead of being used to cancel out the deficit. Then, since $\beta x_{12} \le |E_2|$ and $w_2 = 0$, BS 2 must compensate for an additional $\delta$ units of deficit by discharging $\delta/\alpha$ units of energy from storage. The storage level therefore drops by $\delta/\alpha$. On the other hand, the $\delta$ units of energy can only increase the storage levels by $\alpha \delta$. The net change in storage level is therefore $\alpha \delta -\delta/\alpha \le 0$, which is sub-optimal compared to the case where the $\delta$ units of energy is simply used for canceling the deficit. 
\end{proof}

Let $\Delta \ge 0$ be the amount of energy sent to storage at BS 1. We have the following constraints on $\Delta$.
\begin{align*}
\Delta &\le E_1, \\
\Delta &\le \frac{S_{\rm max} - s_1}{\alpha}, \\
|E_2| -\beta (E_1 - \Delta)  &\le \alpha s_2. 
\end{align*}
The last constraint comes from the fact that we need to transfer enough energy to ensure that $w_2 = 0$. The change in sum storage levels is then given by $\alpha \Delta - (|E_2| - \beta E_1 + \beta \Delta)/\alpha = (\beta E_1 - |E_2|)/\alpha + (\alpha -\beta/\alpha) \Delta$. Hence, if $\beta \ge \alpha^2$, $\Delta = 0$, and if $\beta < \alpha^2$, $\Delta$ takes its maximum possible value. Hence, the policy is given as follow.
\begin{itemize}
\item If $\beta \ge \alpha^2$: 
\begin{align*}
x_{12} &= E_1, \\
d_2 &= (|E_2| - \beta E_1)/\alpha, \\
s_1 &\leftarrow s_1, \\
s_2 &\leftarrow s_2 - d_2.
\end{align*}

\item If $\beta < \alpha^2$:{\allowdisplaybreaks
\begin{align*}
c_1 &= \min\{E_1, (S_{\rm max} - s_1)/\alpha, E_1 -(|E_2|- \alpha s_2)/\beta\},\\
x_{12} &= E_1 -c_1, \\
d_2 &= \max\{0, (|E_2| - \beta x_{12})/\alpha\}, \\
s_1 &\leftarrow s_1 + \alpha c_1, \\
s_2 &\leftarrow s_2 - d_2.
\end{align*}}

\end{itemize} 

\textit{Case 2.3}: $\beta(E_1 - (S_{\rm max}-s_1)/\alpha)\le |E_2| < \beta E_1$. In this case, $w_2$ is again equal to zero. We next note the following observation.
\begin{itemize}
\item Let $\beta x_{12} \le |E_2|$, then no charging of the storage at BS 2 occurs. This observation follows directly from claim~\ref{clm:et}.
\item If $\beta x_{12} > |E_2|$, then charging of storage at BS 2 occurs and the charge is $ \min\{(S_{\rm max} -s_2)/\alpha, \beta x_{12} - |E_2|\}$. This observation follows similar arguments to claim~\ref{clm:et}. It is always more efficient to use the transferred energy to cancel out the deficit than to charge the storage. 
\end{itemize}

Next, note that the following hold true.
\begin{align}
\alpha c_1 &\le S_{\rm max} -s_1, \\
|E_2| - \beta x_{12} &\le \alpha s_2.\label{eq:casegreedy2}
\end{align}
The first inequality follows from the storage constraint at BS 1. The second follows from the fact that the deficit must be canceled out completely by a combination of energy transfer and storage discharge. Now, we assume without loss of generality that $c_1 = E_1 - x_{12}$. That is, any excess energy is transferred to BS 2. This gives us the inequality
\begin{align}
x_{12} \ge E_1 - (S_{\rm max} -s_1)/\alpha. \label{eq:casegreedy} 
\end{align}

The net change in storage level, $\Delta_S$, is given by
\begin{align*}
 \Delta_S = \alpha (E_1 - x_{12}) - \frac{1}{\alpha}(|E_2| - \beta x_{12})\mathbf{1}_{\beta x_{12} \le |E_2|} + \mathbf{1}_{\beta x_{12} > |E_2|}\alpha \min\{(S_{\rm max} -s_2)/\alpha, \beta x_{12} - |E_2|\}, 
\end{align*} 
where $\mathbf{1}_{(.)}$ is the indicator function and $x_{12}$ satisfies~\eqref{eq:casegreedy} and~\eqref{eq:casegreedy2}. Consider now the case where $\beta \ge \alpha^2$. 

If $\beta x_{12} \le |E_2|$, $\Delta_S$ is an increasing function of $x_{12}$ and the maximum increase in storage level is $ \alpha (E_1 - |E_2|/\beta)$. 

On the other hand, if $\beta x_{12} \ge |E_2|$, then we have $x_{12} \ge \max\{|E_2|/\beta,  E_1 - (S_{\rm max} -s_1)/\alpha\}$ and $\Delta_S$ is a decreasing function of $x_{12}$. Hence, the maximum increase in storage level is $\alpha (E_1 - |E_2|/\beta)$ if $E_1 -(S_{\rm max} -s_1)/\alpha \le |E_2|/\beta$ and $S_{\rm max}-s_1+ \alpha \min\{(S_{\rm max} -s_2)/\alpha, \beta(E_1 -(S_{\rm max}-s_1)/\alpha) - |E_2|\}$ otherwise. 

In summary, in this case, we always transfer energy to compensate for all the deficit first before charging storage 1 followed by storage 2. Hence, if $\beta \ge \alpha^2$, the optimal policy is given by{\allowdisplaybreaks
\begin{align*}
x_{12} &= |E_2|/\beta, \\
c_1 &= \min\{(S_{\rm max}-s_1)/\alpha, E_1 - x_{12}\}, \\
s_1 &\leftarrow s_1 + \alpha c_1, \\
c_2 &=\min\{(S_{\rm max} -s_2)/\alpha, \beta (E_1 - x_{12} - c_1)\},\\
x_{12} & \leftarrow x_{12}+ c_2/\beta, \\
s_2 &\leftarrow s_2 + \alpha c_2.
\end{align*}
}
If $\beta < \alpha^2$, then $\Delta_{S}$ is a always a decreasing function of $x_{12}$. Hence, one sets $x_{12}$ as small as possible, subject to~\eqref{eq:casegreedy} and~\eqref{eq:casegreedy2}. 
That is, we set $x_{12} = \max\{(|E_2| - \alpha s_2)/\beta, E_1 - (S_{\rm max} -s_1)/\alpha\}$. In summary, for this case, the optimal policy is given by
\begin{align*}
x_{12} &= \max\{(|E_2|- \alpha s_2)/\beta, E_1- (S_{\rm max} -s_1)/\alpha, 0\}, \\
c_1 &= \min\{(S_{\rm max}-s_1)/\alpha, (E_1 - x_{12})\}, \\
d_2 &= \max\{(|E_2| - \beta x_{12})/\alpha, 0\}, \\
c_2 &= \mathbf{1}_{d_2 = 0}\min\{(\beta x_{12} - |E_2|), (S_{\rm max}-s_2)/\alpha\}, \\
s_2 &\leftarrow s_2 + \alpha c_2 - d_2,\\
s_1 &\leftarrow s_1 + \alpha c_1.
\end{align*}

\textit{Case 2.4}: $|E_2| < \beta(E_1 - (S_{\rm max}-s_1)/\alpha)$. This case is straightforward. The excess energy is enough to address the deficit at BS 2 as well as charge the storage at BS 1 to $S_{\rm max}$. From claim~\ref{clm:et}, we also see that the transferred energy from BS 1 is always used first to compensate for the deficit before charging the storage at BS 2. The optimal policy is then given as follow.{\allowdisplaybreaks
\begin{align*}
c_1 & = (S_{\rm max} -s_1)/\alpha, \\
x_{12} & = E_1 - c_1, \\
c_2 &= \min\{(\beta x_{12} - |E_2|), (S_{\rm max}-s_2)/\alpha\}, \\
s_2 &\leftarrow s_2 + \alpha c_2,\\
s_1 &\leftarrow s_1 + \alpha c_1.
\end{align*}}

Finally, from combining all four sub-cases, if $\beta \ge \alpha^2$, the optimal policy can be reduced to the form stated in Case 2A of Proposition~\ref{prop:algo}. When $\beta < \alpha^2$, the optimal policy reduces to Case 2B of  Proposition~\ref{prop:algo}. 

Next, for Case 3 of Proposition~\ref{prop:algo}:  $E_1 < 0$, $E_2 \ge 0$ and $\beta \ge \alpha^2$. This case is symmetric to Case 2, with the roles of BSs 1 and 2 reversed. We therefore omit theproof here and refer readers to proof for Case 2 above. 

\textit{Case 4}:  $E_1 < 0$ and $E_2 < 0$. In this case, each BS compensates using individual storage first, before helping the other, since $\alpha \ge \alpha \beta$. Hence, it is less efficient to use the storage of the other BS to compensate for deficit if there is still storage in the current BS. Therefore, for $i\in\{1,2\}$, we let 
\begin{align*}
d_i &= \min\{s_i, |E_i|/\alpha\}, \\
E'_i &= E_i + \alpha d_i, \\
s_i &\leftarrow s_i - d_i.
\end{align*}
If either $E'_1 \ge 0$ or $E'_2 \ge 0$, the algorithm reduces to the first three cases with net energy profiles being $E'_1$ and $E'_2$. If both $E'_1<0$ and $E'_2<0$, we compensate with conventional energy generation and set $w_i = |E'_i|$.
\end{proof}
\section{Proof of Proposition~\ref{prop2}} \label{append:prop2}
\begin{proof}
The proof follows from the observation that a system with $s(t) \le s'(t)$ component-wise can mimic the optimal policy of a system at state $s'(t)$ using conventional energy. 

Let $\pi^*(s'(k))$ denote an optimal policy when the state is at $s'(k)$, and $(w^*(t), c^*(t), d^*(t), x^*_{12}(t), x^*_{21}(t))$ denote the control variables induced by the energy profile and optimal policy for $t \ge k$. We also use $s^*(t)$ to denote the evolution of the state under the optimal policy, starting from $s'(k)$. 

Let $\pi(s(k))$ denote a control policy when the state is at $s(k)$, and $(w(t), c(t), d(t), x_{12}(t), x_{21}(t))$ denote the control variables induced by the energy profile and optimal policy for $t \ge k$. We use $s(t)$ to denote the evolution of the state under the policy, starting from $s(k)$. Now, we set $\pi = \pi^*$ except when  $d_i^*(t) > s_i(t)$ for any $i\in \{1,2\}$. In this case, we set 
\begin{align}
d_i(t) &= s_i(t), \label{eq:3}\\
 w_i(t) &= w_i^*(t) + \alpha (d_i^*(t) - s_i(t)). \label{eq:2} 
\end{align}
Observe now that as the optimal policy $\pi^*$ satisfies the energy constraints at each time $t$, $\pi$ also satisfies the energy constraints through compensating for any additional discharge under the optimal policy using conventional energy (the term $\alpha (d_i^*(t) - s_i(t))$ in~\eqref{eq:2}). For the storage constraints, observe that the discharging constraints are taken care of by~\eqref{eq:3} and~\eqref{eq:2}. As for the charging constraints, observe that since $s(k) \le s'(k)$ component-wise, and the discharging policy in~\eqref{eq:3} and~\eqref{eq:2} still results in $s(t) \le s^*(t)$ for all $t\ge k$, the optimal charging policy ($c^*(t)$) can be accommodated under $\pi$. 
Furthermore, we have for $i\in{1,2}$
\begin{align*}
s^*_{i}(t+1) - s_{i}(t+1)& \le s^*_{i}(t) - s_{i}(t) - 1_{d^*_i(t) \ge s_i(t)} (d^*_i(t) - s_i(t)),
\end{align*}
where $1_{\{.\}}$ denotes the indicator function. Since $s^*_{i}(N) - s_{i}(N) \ge 0$, we have 
\begin{align}
\sum_{t = k}^N1_{d^*_i(t) \ge s_i(t)} (d^*_i(t) - s_i(t)) \le s'_{i}(k) - s_{i}(k). \label{eq:4} 
\end{align}
We also have from~\eqref{eq:2}
\begin{align}
 \sum_{t = k}^Nw_i(t) &= \sum_{t = k}^N (w_i^*(t) + \alpha1_{d^*_i(t) \ge s_i(t)} (d_i^*(t) - s_i(t))). \label{eq:5}
\end{align}
\eqref{eq:4} and \eqref{eq:5} implies that
\begin{align*}
&J_{\pi}(s(k)) - J_{\pi^*}(s'(k))\\
 &= \alpha\sum_{t = k}^N \left[\begin{array}{c}1_{d^*_1(t) \ge s_1(t)} \\ 1_{d^*_2(t) \ge s_2(t)}\end{array}\right]^T(d^*(t) - s(t)) \\
& \le \alpha \sum_{i=1}^2(s'_i(k) - s_i(k)).
\end{align*}
Hence,
\begin{align*}
J_{\pi^*}(s(t)) &\le J_{\pi}(s(t))\\
& \le J_{\pi^*}(s'(k)) + \alpha \sum_{i=1}^2(s'_i(k) - s_i(k)).
\end{align*}
\end{proof}

\section{Proof of Proposition~\ref{prop3}} \label{pf:storage}
\begin{proof}
Both cases 1 and 2 of Proposition~\ref{prop5} use the same lines of argument. To avoid repetition, we only give the proof for Case 1. Consider a system at state $s'(k) = [s_1(k),  s_2(k)+ \alpha\beta \Delta,]^T$. Let $\pi^*(s'(k))$ be the optimal policy corresponding to this state. Consider now a system at state $s''(k) = [s_1(k)+ \alpha \Delta,  s_2(k)]^T$. Let $\pi(s''(k)) = \pi^*(s'(k))$ except in the following cases:
\begin{enumerate}
\item Deficit at BS 2: $d_2^*(t) > s_2''(t)$. In this case, the deficit is given by $\alpha (d^*_2(t) -s_2''(t))$. We will compensate for this deficit by transferring energy from storage 1 (at BS 1) to BS 2. The total amount that needs to be transfer out of storage 1 is $\alpha(d_2^*(t) - s_2''(t))/(\alpha\beta)= (d_2^*(t) - s_2''(t))/\beta$.
\item Overcharging at BS 1: $\alpha c_1^*(t) > S_{\rm max} - s_1''(t) $. In this case, we set $\hat{c}_1(t) = (S_{\rm max} - s_1''(t))/\alpha$ to charge storage 1 and transfer $\min\{c^*_1(t) - \hat{c}_1(t), (s'_2(t) - s_2''(t))/(\alpha\beta)\}$ to charge storage 2.
\end{enumerate}
In the first case, the total amount of energy that can be transfer from BS 1 to compensate for any deficit at BS 2 is $\alpha^2 \beta \Delta$, whereas the maximum amount of deficit that we incur is at most $\alpha^2\beta \Delta$. Observe also that in neither cases do we need to use additional energy from the grid to compensate for any deficits.  

In the second case, observe that the total amount of excess energy transferable from BS 1 is $\alpha\beta \Delta/\alpha = \beta \Delta$. This amount of energy leads to an increase in storage level of $\alpha \beta \Delta$. Since the gap between $s'_2(t)$ and $s''_2(t)$ is at most $\alpha\beta \Delta$, this energy transfer policy can be used to compensate for the gap in storage level at storage 2.

More formally, we have the following claim 
\begin{claim}~\label{clm:store}
For evolution of state $s''(t)$ under policy $\pi$ and evolution of $s'(t)$ under $\pi^*$, and the same energy profiles, we have for $t \ge k$, 
\begin{align}
s''_1(t) &\ge s_1'(t), \label{p3:eq1}\\
s_2'(t) &\ge s''_2(t), \label{p3:eq2}\\
s_2'(t) - s''_2(t) &\le \beta(s''_1(t) - s_1'(t)) \label{p3:eq3} 
\end{align}
\end{claim}
\begin{proof}
This set of inequalities are clearly true for $t = k$. We now show by induction that they are also true for all $ t\ge k$. Assume that the set of inequalities are true at time $t$. We now consider the two scenarios listed that result in a change in the difference of storage levels $s''_1(t) - s_1'(t)$ and $s_2'(t) - s''_2(t)$. 

If $d_2^*(t) > s''_2(t)$: The difference is first compensated by storage at BS 1, resulting in a drop in storage level for BS 1. Note that since $s_2'(t) \ge d_2^*(t)$, we have that $s_2'(t) - s_2''(t) \ge d_2^*(t) - s_2''(t)$. By the induction hypothesis, we therefore have $s_1(t) - s_1'(t) \ge  (d_2^*(t) - s_2''(t))/\beta$. Hence, we can compensate for the deficit by discharging from BS 1\footnote{The deficit is at BS 2 $\alpha (d_2^*(t) - s_2''(t))$. We compensate by discharging $(d_2^*(t) - s_2''(t))/\beta$ at storage at BS 1, resulting in a net energy of $\alpha (d_2^*(t) - s_2''(t))$ after storage and energy transfer loss. \label{fn2}}. The drop in storage levels at BS 1 is then given by
\begin{align}
s_1''(t+1) - s_1'(t+1) &= s_1''(t) - s_1'(t) - \frac{(d_2^*(t)- s_2(t))}{\beta}. \label{eq:spf1}
\end{align}
For BS 2, we have that
\begin{align}
s_2'(t+1) - s_2''(t+1) &= s_2'(t) - s_2''(t) - (d_2^*(t)- s_2''(t)). \label{eq:spf2}
\end{align}
Inequalities~\eqref{eq:spf1} and~\eqref{eq:spf2} imply that
\begin{align*}
s_2'(t+1) - s_2''(t+1) &\le \beta(s_1''(t+1) - s_1'(t+1)). 
\end{align*}

Next, consider the case when charging occurs at BS 2 (and no discharging occurs). By the induction hypothesis, $s_2'(t) \ge s_2''(t)$. Hence, $s_2'(t) - s_2''(t) = s_2'(t+1) - s_2''(t+1)$ unless excess charge is transfered over from BS 1. That is, unless $\alpha c_1^*(t) > S_{\rm max} - s_1''(t)$. In that case, we have
\begin{align*}
s_2'(t+1) - s_2''(t+1) &= \max\{0, s_2'(t) - s_2''(t) - \alpha\beta(c^*_1(t) - \hat{c}_1(t))\}. 
\end{align*}
Hence, we have $s_2'(t+1) \ge s_2''(t+1)$ and from the induction hypothesis,
\begin{align*}
s_2'(t+1) - s_2''(t+1) &\le \max\{0, \beta(s_1''(t) - s_1'(t)) - \alpha\beta(c^*_1(t) - \hat{c}_1(t))\} \\
& =  \beta(s_1''(t) - s_1'(t)) - \alpha\beta(c^*_1(t) - \hat{c}_1(t)) \\
& = \beta (s_1''(t+1) - s_1'(t+1)).
\end{align*}
\end{proof}
From~\eqref{p3:eq3}, it is clear that $J_{\pi}(s''(k)) \le J_{\pi^*}(s'(k))$ since any additional deficit that occurs at BS 2 can be compensated for through discharge from storage of BS 1 (see footnote~\ref{fn2}). Hence, we have
\begin{align*}
J_{\pi^*}(s''(k)) & \le J_{\pi}(s''(k))\\
&= J_{\pi}\left(\left[\begin{array}{c}s_1(k) + \alpha \Delta \\ s_2(k)\end{array}\right]\right) \\
&\le J_{\pi^*}(s'(k)) \\
& = J_{\pi^*}\left(\left[\begin{array}{c}s_1(k) \\ s_2(k) + \alpha \beta \Delta\end{array}\right]\right),
\end{align*}
which completes the proof of this Proposition. 
\end{proof}
\section{Proof of Proposition~\ref{prop4}} \label{pf:transfer}
\begin{proof}
We note that at time $t = N$, transferring $\min\{E_1(N), |E_2(N)|/\beta\}$ units of energy from BS 1 to BS 2 first is optimal. Hence, it remains to show that, for $t<N$, the cost to go function for a policy that transfers $\min\{E_1(t), |E_2(t)|/\beta\}$ to compensate for $E_2(t)$ is optimal. We first note that for the energy sent to BS 2, it is optimal to use all transferred energy to compensate for the deficit $E_2(k)$ first, rather than sending the energy to storage. This follows from Proposition~\ref{prop2}. Let $\Delta_2$ be the energy at BS 2 that comes from BS 1, and let $\Delta_{S2} \le \Delta_2$ be the part of the energy that is sent to storage at BS 2 instead of being used to compensate for the deficit. If $\Delta_{S2} >0$, then charging of storage 2 occurs and we can assume without loss of generality that no discharging occurs. Hence, the deficit $|E_2(k)|$ must be compensated for by the remaining transfer energy, $\Delta_2 - \Delta_{S2}$, and energy drawn from the main grid, $w_2^*(k)$. From Corollary~\ref{coro1},  we can assume that $w_2^*(k)$ is not used to charge the storage at BS 2. Hence, we have $w_2^*(k) = \max\{|E_2|- \Delta_2 + \Delta_{S2}, 0\}$. If $\Delta_{2} -\Delta_{S2} \ge |E_2(k)|$, then it means that we compensate for the deficit first before charging the storage. If $\Delta_{2} -\Delta_{S2} < |E_2(k)|$, then the optimal cost to go function is lower bounded by
{\allowdisplaybreaks
\begin{align*}
J_{\pi^*}\left(\left[\begin{array}{c}s_1(k)\\ s_2(k)\end{array}\right]\right) &= w_1^*(k) +|E_2|- \Delta_2 + \Delta_{S2} + J_{\pi^*}\left(\left[\begin{array}{c}s_1(k+1)\\ s_2(k) + \alpha \Delta_{S2}\end{array}\right]\right) \\
& \ge w_1^*(k) +|E_2|- \Delta_2 + \Delta_{S2} + J_{\pi^*}\left(\left[\begin{array}{c}s_1(k+1)\\ s_2(k)\end{array}\right]\right) -\alpha^2 \Delta_{S2},
\end{align*}}
where the second inequality follows from Proposition~\ref{prop2}. Since $(1-\alpha^2)\Delta_2 \ge 0$, the optimal $\Delta_{S2}$ is given by $\Delta_{S2} = \max\{\Delta_2 - |E_2(k)|, 0\}$ for $\Delta_2 - \Delta_{S2}\le |E_2(t)|$, which corresponds to using all of the transferred energy to cancel out the deficit first before charging the storage at BS 2.

Next, since all of the transferred energy $\Delta_2$, is used to compensate for the deficit first, if $\Delta_2 \ge |E_2(k)|$, then the Proposition is proven. If $\Delta_2 = \beta E_1(k)$, the Proposition is also proven since all of the excess energy at BS 1 is transferred to BS 2. Hence, it remains to consider the case where $\Delta_2 < \min\{\beta E_1(k), |E_2(k)|\}$. Here, a part of $E_1(k)$, $E_1(k) - \Delta_2/\beta$, is sent to storage at BS 1 instead of being transferred to BS 2 to compensate for $E_2(t)$. We can assume without loss of generality that there is no discharging of storage 1 since charging occurs. Let $\Delta =\min\{\beta E_1(k), |E_2(k)|\} - \Delta_2$. The deficit of $\beta \Delta$ at BS 2 has to be compensated for by other means, either through conventional power generation or through energy drawn from storage at BS 2. We consider the two cases separately



Case 1: Energy drawn from storage 2. Then, in this case, we show that
\begin{align*}
J_{\pi^*}\left(\left[\begin{array}{c}s_1(k) \\ s_2(k)\end{array}\right]\right)  \le J_{\pi^*}\left(\left[\begin{array}{c}s_1(k) + \alpha \Delta \\ s_2(k) - \beta\Delta/\alpha\end{array}\right]\right). 
\end{align*}
The proof follows similar arguments for Proposition~\ref{prop3} (see Appendix~\ref{pf:storage}). Consider a system at state $s'(k) = [s_1(k) + \alpha \Delta,  s_2(k) - \beta\Delta/\alpha]^T$. Let $\pi^*(s'(k))$ be the optimal policy corresponding to this state. Consider now a system at state $s(k) = [s_1(k),  s_2(k)]^T$. Let $\pi(s(k)) = \pi^*$ except in the following cases:
\begin{enumerate}
\item Overcharging at BS 2: $\alpha c_2^*(t) > S_{\rm max} - s_2(t) $. In this case, we set $\hat{c}_2(t) = (S_{\rm max} - s_2(t))/\alpha$ to charge storage 2 and transfer $\min\{c^*_2(t) - \hat{c}_2(t), (s'_1(t) - s_1(t))/(\alpha\beta)\}$ to charge storage 1.
\item Deficit at BS 2: $d_1^*(t) > s_1(t)$. In this case, the deficit is given by $\alpha (d^*_1(t) -s_1(t))$. We will compensate for this deficit by transferring energy from storage 2 (at BS 2) to BS 1. The total amount that needs to be transfer out of storage 1 is $\alpha(d_1^*(t) - s_1(t))/(\alpha\beta)= (d_1^*(t) - s_1(t))/\beta$.
\end{enumerate}

Similar to Claim~\ref{clm:store}, for evolution of state $s(t)$ under policy $\pi$ and evolution of $s'(t)$ under $\pi^*$, and the same energy profiles, we have for $t \ge k$, 
\begin{align}
s_1'(t) &\ge s_1(t),\label{eq:t1}\\
s_2(t) &\ge s_2'(t), \label{eq:t2}\\
s_1'(t) - s_1(t) &\le \beta(s_2(t) - s_2'(t)).\label{eq:t3} 
\end{align}
As the proof for these inequalities follow the same arguments as those found in Claim~\ref{clm:store} in Appendix~\ref{pf:storage}, we omit the proof here. We note only that the condition $\beta > \alpha$ is required for the inequalities to hold at $t = k$. At $t = k$, $\beta (s_2(t) - s_2'(t)) = \beta^2 \Delta/\alpha$, while $s_1'(t) -s_1(t) = \alpha\Delta$. Hence, if $\beta > \alpha$, $\beta^2 \Delta/\alpha > \alpha \Delta$.

When inequalities~\eqref{eq:t1} to~\eqref{eq:t3} are satisfied, the policy $\pi$ starting at state $s(k)$ does not incur more energy cost than the optimal policy $\pi^*$ for state $s'(k)$. Hence, we have
\begin{align*}
J_{\pi^*}\left(\left[\begin{array}{c}s_1(k) \\ s_2(k)\end{array}\right]\right)  &\le J_{\pi}\left(\left[\begin{array}{c}s_1(k) \\ s_2(k)\end{array}\right]\right)\\
& \le J_{\pi^*}\left(\left[\begin{array}{c}s_1(k) + \alpha \Delta \\ s_2(k) - \beta\Delta/\alpha\end{array}\right]\right), 
\end{align*}
which implies that the optimal $\Delta =0$.



Case 2: Increase in conventional energy. In this case, we incur an additional loss of $\beta \Delta$. Then, we have
\begin{align*}
J_{\pi^*}\left(\left[\begin{array}{c}s_1(k) \\ s_2(k)\end{array}\right]\right) &\le J_{\pi^*}\left(\left[\begin{array}{c}s_1(k) + \alpha \Delta \\ s_2(k)\end{array}\right]\right) - \beta \Delta\\
&\le -\beta \Delta + \alpha^2 \Delta + J_{\pi^*}\left(\left[\begin{array}{c}s_1(k) \\ s_2(k)\end{array}\right]\right),
\end{align*}
where the last line follows from Proposition~\ref{prop2}. Since $\beta > \alpha \ge \alpha^2$, we have $\Delta \le 0$, which implies that the optimal $\Delta =0$.
\end{proof}
\section{Proof of Proposition~\ref{prop5}} \label{append:prop5}

\begin{proof} At time $t = N$, it is clear that the greedy policy minimizes the amount of energy drawn from the grid. It remains to show for $t < N$ that the greedy policy, $\pi_g$, minimizes the cost-to-go function. That is, $J_{\pi_g} = J_{\pi^*}$. We will do so using a backward induction argument. Assume that at time $t = k$, we follow the greedy policy and then revert back to the optimal policy at time $k+1$. We show that this one-step greedy approach is also an optimal policy. Since the greedy policy is optimal at time $t = N$, induction on $t$ then shows that the greedy policy is optimal for all $t$. Let $\pi_{og}$ denote the one step greedy policy. Then,
\begin{align*}
J_{\pi_{og}}\left(\left[\begin{array}{c}s_1(k) \\ s_2(k)\end{array}\right]\right) = \mathbf{1}^T w_g(k) + J_{\pi^*}\left(\left[\begin{array}{c}s_{g,1}(k+1) \\ s_{g,2}(k+1)\end{array}\right]\right),
\end{align*}
where $w_g(t)$ represents the energy drawn from the main grid under the greedy policy at $t = k$, and $s_g(k+1)$ represents the storage states at time $k+1$ after applying the greedy policy at time $k$. Note that the greedy policy is designed to minimize the conventional energy drawn from the grid at time $k$. Hence, even under the optimal policy, $\pi^*$, we have $\mathbf{1}^T w_g(t) \le \mathbf{1}^T w^*(t)$. At each time $t = k$, there are two cases to consider.

Case 1: $E_2(k) \ge 0$. In this case, $\mathbf{1}^T w_g(k) = 0$ and each BS charges its own storage first before charging the storage of the other BS. In this case, it is straightforward to see from Case 1 of Proposition~\ref{prop3} and Corollary~\ref{coro1} that $J_{\pi_{og}} = J_{\pi^*}$. 

Case 2.1: $E_2(k) < 0$ and $E_1(k)\ge |E_2(k)|/\beta$. In this case, from Proposition~\ref{prop4} and the condition $\beta > \alpha$, energy transfer from BS 1 to BS 2 is an optimal strategy. Since $E_1(t)\ge |E_2(t)|/\beta$, we can reduce the problem back to the first case with $E_1'(t) = E_1(t) - |E_2(t)|/\beta$ and $E_2'(t) =0$, where the greedy strategy is optimal. 


Case 2.2: $E_2(k) < 0$ and $E_1(k)< |E_2(k)|/\beta$. From Proposition~\ref{prop4}, energy transfer at time $k$ is still optimal. Hence, we have 
\begin{align*}
J_{\pi_{og}}(s(k), E_1(k), E_2(k)) = J_{\pi_{og}}(s(k), 0, E_2(k)+ E_1(k)/\alpha), \\
J_{\pi^*}(s(k), E_1(k), E_2(k)) = J_{\pi^*}(s(k), 0, E_2(k)+ E_1(k)/\alpha). 
\end{align*}

It remains to show that $J_{\pi_{og}}(s(k), 0, E_2(k)+ E_1(k)/\alpha) \le J_{\pi^*}(s(k), 0, E_2(k)+ E_1(k)/\alpha)$. Let $\pi'$ be any other policy. Since $\pi_{og}$ minimizes the conventional energy required at time $k$, $\mathbf{1}^T(w'(k) - w_g(k))$. Note from Corollary~\ref{coro1} that conventional energy is not used to charge storages 1 or 2. Since the deficit occurs at BS 2, we can assume without loss of generality that $w_1'(k) = w_{g}(k) = 0$ and $\Delta = w_2'(k) - w_{g,2}(k)$. Let $\Delta_1$ and $\Delta_2$ be the change in storage levels, with respect to the greedy policy, due to the policy $\pi'$. We have
\begin{align*}
J_{\pi'}\left(\left[\begin{array}{c}s_1(k) \\ s_2(k)\end{array}\right]\right) \ge w'_2(k) + J_{\pi^*}\left(\left[\begin{array}{c}s_{g,1}(k+1) + \Delta_1 \\ s_{g,2}(k+1) + \Delta_2\end{array}\right]\right).
\end{align*}
Now, $\Delta_2 \ge 0$ since $\pi_{og}$ first uses storage 2 to compensate for any deficit before using storage 1 and conventional energy. Further, from Corollary~\ref{coro1}, the conventional energy is not used to charge the storages. Therefore, any change in storage levels is due to the additional conventional energy, $\Delta$, being used to compensate for the deficit instead of storage discharges. Hence, $\alpha \Delta_2 + \alpha \beta \Delta_1 = \Delta$. Now, if $\Delta_1 < 0$, we have{\allowdisplaybreaks
\begin{align*}
J_{\pi'}\left(\left[\begin{array}{c}s_1(k) \\ s_2(k)\end{array}\right]\right) &\ge w'_2(k) + J_{\pi^*}\left(\left[\begin{array}{c}s_{g,1}(k+1) + \Delta_1 \\ s_{g,2}(k+1) + \Delta_2\end{array}\right]\right) \\
& \stackrel{(a)}{\ge}w'_2(k) + J_{\pi^*}\left(\left[\begin{array}{c}s_{g,1}(k+1) + \Delta_1+ \beta\Delta_2 \\ s_{g,2}(k+1) \end{array}\right]\right) \\
& \stackrel{(b)}{\ge}w'_2(k)-w_{g,2}(k)+ w_{g,2}(k) + J_{\pi^*}\left(\left[\begin{array}{c}s_{g,1}(k+1)  \\ s_{g,2}(k+1) \end{array}\right]\right)-  \alpha\Delta_1- \alpha\beta\Delta_2 \\
& = w_{g,2}(k) + J_{\pi^*}\left(\left[\begin{array}{c}s_{g,1}(k+1)  \\ s_{g,2}(k+1) \end{array}\right]\right)+ \Delta -(\alpha\Delta_1+ \alpha\beta\Delta_2) \\
& = J_{\pi_{og}}\left(\left[\begin{array}{c}s_{1}(k)  \\ s_{2}(k) \end{array}\right]\right).
\end{align*}}
$(a)$ follows from Case 2 of Proposition~\ref{prop3}, $\Delta_1 \le 0$ and $\Delta_2 + \beta \Delta_1 = \Delta/\alpha \ge 0$. $(b)$ follows from Proposition~\ref{prop2}. Now, for the case when $\Delta_1 \ge 0$, we have
\begin{align*}
J_{\pi'}\left(\left[\begin{array}{c}s_1(k) \\ s_2(k)\end{array}\right]\right) &\ge w'_2(k) + J_{\pi^*}\left(\left[\begin{array}{c}s_{g,1}(k+1) + \Delta_1 \\ s_{g,2}(k+1) + \Delta_2\end{array}\right]\right) \\
& \stackrel{(a)}{\ge}w'_2(k)-w_{g,2}(k)+ w_{g,2}(k) + J_{\pi^*}\left(\left[\begin{array}{c}s_{g,1}(k+1) \\ s_{g,2}(k+1) + \Delta_2 \end{array}\right]\right) - \alpha \Delta_1\\
& \stackrel{(b)}{\ge}\Delta+ w_{g,2}(k) + J_{\pi^*}\left(\left[\begin{array}{c}s_{g,1}(k+1)  \\ s_{g,2}(k+1) \end{array}\right]\right)-  \alpha\Delta_1- \alpha\beta\Delta_2 \\
& = w_{g,2}(k) + J_{\pi^*}\left(\left[\begin{array}{c}s_{g,1}(k+1)  \\ s_{g,2}(k+1) \end{array}\right]\right)+ \Delta -(\alpha\Delta_1+ \alpha\beta\Delta_2) \\
& = J_{\pi_{og}}\left(\left[\begin{array}{c}s_{1}(k)  \\ s_{2}(k) \end{array}\right]\right).
\end{align*}
$(a)$ follows from Proposition~\ref{prop2}. $(b)$ follows from claim~\ref{clm:bound}. 
\end{proof}

\section{Proof of Proposition~\ref{prop6}} \label{append:prop6}
\begin{proof}
We first note that the modified greedy policy is optimal at time $t = N$. It now remains to show by backward induction that the policy is optimal for all $t$. As with Proposition~\ref{prop5}, let $\pi_{og}$ denote the one step modified greedy policy in which the Case 2A of Proposition~\ref{prop5} is implemented at time $t = k$ and then the optimal policy is implemented for $t \ge k+1$. We now show that any other policy, $\pi'$, will incur a cost that is at least as large as the cost incurred by $\pi_{og}$. 

Observe that for a policy $\pi'$ to be different from $\pi_{og}$ at time $t = k$, the energy transferred to BS 2, $x_{12}(k)$, must be less than $\min\{|E_2(k)|/\beta, E_1(k)\}$. That is, a fraction of the excess energy is put into storage instead of being sent to BS 2. Since $x_{12}(k)< \min\{|E_2(k)|/\beta, E_1(k)\}$, the deficit $\Delta = |E_2(k)| - \beta x_{12}(k)$ must be compensated for by other means, through discharging of storage at BS 2 and conventional energy. Let $\Delta_{d2}$ be the additional discharge at storage 2 and $\Delta_2$ be the additional conventional energy (with respect to the modified greedy policy) used to compensate for the deficit. Let $\Delta_{c1} = E_1(k) - x_{12}(k)$ be the additional energy sent to storage 1 (with respect to the modified greedy policy and $\alpha \Delta_{c1} \le S_{\rm max} - s_{g,1}(k+1)$), such that $\Delta_2 + \alpha \Delta_{d2} = \Delta$ and $\beta\Delta_{c1} = \Delta$. 

Case 1: We first consider the case where $\Delta/\alpha \le s_{g,2}(k+1)$. 
{\allowdisplaybreaks
\begin{align*}
J_{\pi'}\left(\left[\begin{array}{c}s_1(k) \\ s_2(k)\end{array}\right]\right) &\ge w_{g,2}(k) + \Delta_{2} + J_{\pi^*}\left(\left[\begin{array}{c}s_{g,1}(k+1) + \alpha\Delta_{c1} \\ s_{g,2}(k+1) - \Delta_{d2}\end{array}\right]\right) \\
&\stackrel{(a)}{\ge} w_{g,2}(k)  + J_{\pi^*}\left(\left[\begin{array}{c}s_{g,1}(k+1) + \alpha\Delta_{c1} \\ s_{g,2}(k+1)- \Delta_{d2} - \Delta_2/\alpha \end{array}\right]\right)  \\
& = w_{g,2}(k)  + J_{\pi^*}\left(\left[\begin{array}{c}s_{g,1}(k+1) + \alpha\Delta_{c1} \\ s_{g,2}(k+1)- \Delta/\alpha\end{array}\right]\right)  \\
& = w_{g,2}(k)  + J_{\pi^*}\left(\left[\begin{array}{c}s_{g,1}(k+1) + \alpha\Delta_{c1} \\ s_{g,2}(k+1)- \beta\Delta_{c1}/\alpha\end{array}\right]\right) \\
& \stackrel{(b)}{\ge}  w_{g,2}(k)  + J_{\pi^*}\left(\left[\begin{array}{c}s_{g,1}(k+1) \\ s_{g,2}(k+1)\end{array}\right]\right) \\
& = J_{\pi_{og}}\left(\left[\begin{array}{c}s_{g,1}(k) \\ s_{g,2}(k)\end{array}\right]\right).
\end{align*} }
$(a)$ follows from Proposition~\ref{prop2}. In $(b)$, we use the following claim. 
\begin{claim}\label{clm:2}
With the assumptions as given in Proposition~\ref{prop6}, for any $k+1 \le t\le N$ and $\Delta \ge 0$ such that $\Delta \le \min\{(S_{\rm max} - s_1(t))/\alpha, \alpha s_2(t)/\beta\}$, we have 
\begin{align*}
J_{\pi^*}\left(\left[\begin{array}{c}s_1(k+1) \\ s_2(k+1)\end{array}\right]\right) \le  J_{\pi^*}\left(\left[\begin{array}{c}s_{1}(k+1) + \alpha\Delta \\ s_{2}(k+1)- \beta\Delta/\alpha\end{array}\right]\right).
\end{align*}
\end{claim}

\textit{Proof of claim}: Let $\pi^*$ be the optimal policy for the system starting at state $[s^*_1(k+1), s^*_2(k+1)]$, where $s^*_1(k+1) = s_{1}(k+1) + \alpha\Delta$ and $s^*_2(k+1) = s_{g,2}(k+1) - \beta\Delta/\alpha$. Let $\pi'$ be a policy for the system starting at state $[s'_1(k+1), s'_2(k+1)]=[s_{1}(k+1), s_{2}(k+1)]$, , such that $\pi' = \pi^*$ except when 
\begin{itemize}
\item $\alpha c_2^*(t) \ge S_{\rm max} - s_2'(t)$: Set $c'_2(t) = (S_{\rm max} - s_2'(t))/\alpha$. Note that since we assume that $E_2(t) \le 0$ for all $t$, and from Corollary~\ref{coro1}, conventional energy is not used to charge storages, any charging of storage at BS 2 must come from the excess energy at BS 1. Hence, we set $x_{12}'(t) = x_{12}^*(t) - (c_2^*(t) - c_2'(t))/\beta$. Finally, we set $c_1'(t) = \min\{c_1^*(t) + (c_2^*(t) - c_2'(t))/\beta, (S_{\rm max} -s_1'(t))/\alpha\}$.
\item  $\alpha d_1^*(t) \ge s_1'(t)$: In this case, note that since $E_1(t) \ge 0$, any discharge from storage at BS 1 is only used at BS 2. We set $d_1'(t) = s_1'(t)$, $x_{12}'(t) = x_{12}^*(t) - \alpha (d_1^*(t) - d_1'(t))$ and $d_2'(t) = d_2^*(t) +\beta(d_1^*(t) - d_1'(t))$. 
\item $c_1^*(t) \ge (S_{\rm max} - s_1''(t))/\alpha$: Set $c_1'(t) = (S_{\rm max} - s_1''(t))/\alpha$. 
\end{itemize}
Using the assumptions $E_1(t) \ge 0$ and $E_2(t) \le 0$, it is not difficult to see that the following inequality regarding the states hold for all $t \ge k+1$.
\begin{align*}
s_2'(t) - s_2^*(t) \ge \beta (s_1^*(t) - s_1'(t))
\end{align*}
if $s_1^*(t)\ge s_1'(t)$ and 
\begin{align*}
s_2'(t) \ge s_2^*(t)
\end{align*}
if $s_1'(t) \ge s_1^*(t)$. These inequalities imply that no additional conventional energy is required when we use policy $\pi'$ for a system starting at state $[s'_1(k+1), s'_2(k+1)]$, as compared to a system under policy $\pi^*$ and starting at state $[s^*_1(k+1), s^*_2(k+1)]$. Hence,
\begin{align*}
J_{\pi^*}\left(\left[\begin{array}{c}s_{1}(k+1) \\ s_{2}(k+1)\end{array}\right]\right) & \le J_{\pi'}\left(\left[\begin{array}{c}s_{1}(k+1) \\ s_{2}(k+1)\end{array}\right]\right) \\
& = J_{\pi'}\left(\left[\begin{array}{c}s_1'(k+1) \\ s_2'(k+1)\end{array}\right]\right) \\
& \le J_{\pi^*}\left(\left[\begin{array}{c}s_{1}(k+1) + \alpha\Delta \\ s_{2}(k+1)- \beta\Delta/\alpha\end{array}\right]\right),
\end{align*}
which completes the proof of Claim~\ref{clm:2}.

Case 2: We now consider the case where $\Delta/\alpha > s_{g,2}(k+1)$. This case can actually be treated as an extension of the previous case (Case 1). Let $\Delta_{c1}' + \Delta_{c2}'' = \Delta_{c1}$ and $\Delta' + \Delta'' = \Delta$ such that $\Delta' = \alpha s_{g,2}(k+1)$ and $\beta \Delta_{c1}' = \Delta'$.{\allowdisplaybreaks
\begin{align*}
J_{\pi'}\left(\left[\begin{array}{c}s_1(k) \\ s_2(k)\end{array}\right]\right) &\ge w_{g,2}(k) + \Delta_{2} + J_{\pi^*}\left(\left[\begin{array}{c}s_{g,1}(k+1) + \alpha\Delta_{c1} \\ s_{g,2}(k+1) - \Delta_{d2}\end{array}\right]\right) \\
&\stackrel{(a)}{\ge} w_{g,2}(k)  +\Delta'' +  J_{\pi^*}\left(\left[\begin{array}{c}s_{g,1}(k+1) + \alpha\Delta_{c1}' + \alpha \Delta_{c1}'' \\ s_{g,2}(k+1)-\Delta'/\alpha \end{array}\right]\right)  \\
&= w_{g,2}(k)  +\Delta'' +  J_{\pi^*}\left(\left[\begin{array}{c}s_{g,1}(k+1) + \alpha\Delta_{c1}' + \alpha \Delta_{c1}'' \\ s_{g,2}(k+1)-\beta\Delta_{c1}'/\alpha \end{array}\right]\right)  \\
&\stackrel{(b)}{\ge} w_{g,2}(k)  +\Delta'' +  J_{\pi^*}\left(\left[\begin{array}{c}s_{g,1}(k+1) + \alpha \Delta_{c1}'' \\ s_{g,2}(k+1) \end{array}\right]\right)  \\
&\stackrel{(c)}{\ge} w_{g,2}(k)  +\Delta'' +  J_{\pi^*}\left(\left[\begin{array}{c}s_{g,1}(k+1)  \\ s_{g,2}(k+1) \end{array}\right]\right) -  \alpha^2\beta \Delta_{c1}''  \\
&{=} w_{g,2}(k)  +\beta (1- \alpha^2)\Delta_{c1}'' +  J_{\pi^*}\left(\left[\begin{array}{c}s_{g,1}(k+1)  \\ s_{g,2}(k+1) \end{array}\right]\right)   \\
& \ge J_{\pi_{og}}\left(\left[\begin{array}{c}s_{g,1}(k) \\ s_{g,2}(k)\end{array}\right]\right).
\end{align*} }
$(a)$ follows from Proposition~\ref{prop2}. $(b)$ follows from claim~\ref{clm:2} and finally, $(c)$ follows from claim~\ref{clm:bound}. 

Combining the two cases then completes the proof of this Proposition. 
\end{proof} 
\bibliographystyle{IEEEtran}
\bibliography{sg}
\end{document}